\documentclass[ASNA]{USG}
\usepackage{amsmath,amssymb,amsfonts,amsthm}
\usepackage{algorithm}
\usepackage{algpseudocode}
\usepackage{complexity}
\usepackage{xspace}
\usepackage{aliascnt}
\usepackage{array}
\usepackage{adjustbox}
\usepackage{float}

\floatname{algorithm}{Algorithm}

\algtext*{EndIf}
\algtext*{EndWhile}
\algtext*{EndFor}

\DeclareMathOperator{\N}{\mathbb{N}}
\DeclareMathOperator{\FDDS}{\mathbb{D}}
\DeclareMathOperator{\treeProd}{\otimes}

\newcommand{\unroll}[1]{\mathcal{U}(#1)}
\newcommand{\tree}[1]{\mathbf{\lowercase{#1}}}
\newcommand{\forest}[1]{\mathbf{\uppercase{#1}}}
\newcommand{\depth}[2][]{d_{#1}(#2)}
\newcommand{\full}[1]{\mathsf{#1}}
\newcommand{\cut}[2]{\mathcal{C}_{#2}(#1)}
\newcommand{\cutU}[2]{\cut{\unroll{#1}}{#2}}
\newcommand{\cutUn}[1]{\cutU{#1}{n}}
\newcommand{\cutT}[2]{\cut{\tree{#1}}{#2}}
\newcommand{\dt}[1]{\mathcal{D}(#1)}
\newcommand{\ie}{\emph{i.e.}\@\xspace}
\newcommand{\cycle}[1]{C_{{#1}}}

\newcommand{\Semiring}[1]{\mathbb{#1}}
\newcommand{\Partial}{\Semiring{P}}
\newcommand{\Nilpotent}{{\Semiring{F}_{f}}}
\newcommand{\Op}[1]{\mathcal{#1}}
\newcommand{\Di}[2]{\Op{D}^{#1}(#2)}
\newcommand{\unrollP}[1]{\mathfrak{U}(#1)}
\newcommand{\nilPart}[1]{\mathbf{F}(#1)}
\newcommand{\fullPart}[1]{\mathsf{T}(#1)}

\newtheorem{Theorem}{Theorem}

\newtheorem{Lemma}[Theorem]{Lemma}
\newtheorem{Corollary}[Theorem]{Corollary}

\newtheorem{Example}[Theorem]{Example}

\newtheorem{prop}[Theorem]{Proposition}

\makeatletter
\let\c@algorithm\relax
\let\c@figure\relax
\let\c@table\relax
\makeatother

\newaliascnt{algorithm}{Theorem}
\newaliascnt{figure}{Theorem}
\newaliascnt{table}{Theorem}

\newtheorem{Definition}[Theorem]{Definition}

\graphicspath{{images/}}

\articletype{RESEARCH ARTICLE}
\journal{...}
\volume{...}
\copyyear{...}
\startpage{0}
\articledoi{...}

\begin{document}

\title{Solving polynomial equations over partial discrete dynamical systems}
\titlemark{Solving polynomial equations over partial discrete dynamical systems}

\author[1]{Maximilien Gadouleau}[https://orcid.org/0000-0003-4701-738X]
\author[2]{Sara Riva}[https://orcid.org/0000-0003-2133-8089]
\author[3]{Marius Rolland}[https://orcid.org/0009-0006-9073-9803]
\author[2]{Marie-Emile Voge}[https://orcid.org/0000-0002-1361-9094]

\authormark{Gadouleau \textsc{et al.}}

\address[1]{\orgdiv{Department of Computer Science,}\orgname{Durham University,}%
\orgaddress{\state{Durham,}\country{United Kingdom}}}

\address[2]{\orgdiv{}\orgname{Universit\'e de Lille, CNRS, Inria, UMR 9189 CRIStAL,}%
\orgaddress{\state{F-59000 Lille,}\country{France}}}

\address[3]{\orgdiv{Aix-Marseille Universit\'e, CNRS, LIS,}\orgname{}%
\orgaddress{\state{Marseille,}\country{France}}}

\corres{Marius Rolland (\email{marius.rolland@lis-lab.fr})}

\abstract[ABSTRACT]{The analysis of observable phenomena (for instance, in biology or physics) allows the detection of dynamical behaviours. 
	If the conditions are ideal and the number of observations is sufficient, we can represent these phenomena by a dynamical system, also called a functional digraph, that is to say a graph where each node has out-degree exactly one.
	Up to isomorphism, these dynamical systems, equipped with disjoint union as addition and direct product as multiplication, form a commutative semiring. 
	Several previous studies on this semiring have aimed to establish algebraic properties (primality, injectivity) or complexity results (division, factorisation). 
	However, no work has yet been conducted on graphs derived from imperfect observations that result in missing transitions or nodes, in other words, in cases where each node in the graph has an out-degree of at most one. In this case, we say that the system is partial. 
	In this paper, we show that partial dynamical systems, up to isomorphism and equipped with the same addition and multiplication, still form a commutative semiring. 
	We then characterise the prime elements of this semiring, which differ from those of the semiring of dynamical systems. 
	Finally, we highlight two properties shared by both semirings. First, injective univariate polynomials admit the same characterisation in both. 
	Second, division can be computed in polynomial time for partial dynamical systems if and only if it can be for dynamical systems.}

\keywords{Finite dynamical systems | Functional digraphs | Direct product | Polynomial equations | Semiring}

\maketitle
	\section{Preliminaries}

A \textbf{Finite Deterministic Dynamical System} (FDDS) can be viewed in full generality as a pair $X = (C_X, f_X)$ where $C_X$ is a set of states and $f_X$ is a total function from $C_X$ to $C_X$, called the \emph{transition function}.
Such objects can be encoded by their \emph{transition graph}, that is, a graph $G_X$ with vertex set $C_X$ and containing an arc from a vertex $u$ to a vertex $v$ if and only if $f_x(u) = v$.
These graphs are also known as functional digraphs.
We denote by $\FDDS$ the set of FDDS taken up to isomorphism.
In the following, the term FDDS will always refer to an element of $\FDDS$.

In the graph representation, an FDDS can be viewed as a multiset of \emph{weakly connected components} (hereafter simply called connected components), each consisting of a unique limit cycle encoding the periodic states of the system, and of \emph{finite rooted trees} (hereafter simply called trees) attached at cycle vertices.
Among connected FDDS, we distinguish two particular types: cycles and dendrons.
Cycles are connected components with no trees, and the cycle of length $n$ is denoted $\cycle{n}$; FDDS consisting only of cycles are called permutations.
Dendrons are connected components whose cycle has length $1$.

It is shown in~\cite{dorigatti2018polynomial} that $\FDDS$ equipped with the disjoint union (denoted $+$) of transition graphs as addition and the direct product (denoted $\times$) of transition graphs as multiplication forms a commutative semiring.
Since every semiring gives rise to a polynomial semiring, we can introduce polynomial equations over $\FDDS$.
It is proved in~\cite{dorigatti2018polynomial} that solving polynomial equations is undecidable in the general case.
However, if one side is fixed as a constant — by convention the right-hand side — then the associated decision problem is in the class $\NP$, since the size of the right-hand side (its number of states) bounds the size of solutions and there are only a finite number of FDDS having a bounded size.

The literature has focused on equations of the form $P(X) = B$ where $P$ is a univariate polynomial over FDDS and $B$ is an FDDS\footnote{Throughout this paper, all polynomials will be assumed to be univariate. }.
In particular, equations of the form $AX^k = B$ and $AX = B$ have been widely studied since they can be seen as the starting points to be able to study more complex equations.
Indeed, it has been shown that identifying solutions to some particular equations can be done in polynomial time over the size of the inputs ~\cite{article_arbre,poly_inj,automata2024ext,automata2025,Antonio2026,divisions_par_premier}. 
Note that, in the following, for the sake of clarity we will simply say polynomial time instead of polynomial time over the size of the inputs and solving an equation instead of identifying a solution.

An important turning point to tackle these problems has been the introduction of the notion of unroll introduced in~\cite{article_arbre}.

\begin{Definition}[Unroll]
	Let $A$ be an FDDS and $u$ a vertex of $A$.
	The \emph{depth} of $u$ in $A$, simply denoted $\depth{u}$, is the smallest distance between $u$ and a vertex of a cycle of $A$.

	The \emph{unroll} of $A$, denoted $\unroll{A}$, is the graph whose vertex set consists of pairs $(u, \depth{u} + k)$ with $k \ge 0$ an integer, and containing an arc $\big((u,d),(v,d')\big)$ if and only if $f_A(u) = v$ and $d' = d - 1$.
\end{Definition}

From this definition, we see that $\unroll{A}$ consists of several \emph{unroll trees}, one per vertex in the cycles of $A$.
Each such tree contains a unique infinite branch, arising from the cycle vertices, along which finite trees — the trees rooted at cycle vertices in $A$ — are periodically attached.

We denote by $\unroll{\FDDS}$ the set of unrolls taken up to isomorphism; this convention will be maintained throughout unless stated otherwise.
We can then define a sum of unrolls, denoted $+$, which is again the disjoint union, as well as a product of unrolls, denoted $\treeProd$, defined as follows.
For readability, trees are denoted by a bold lowercase letter, e.g., $\tree{t}$, while forests are denoted by a bold uppercase letter, e.g., $\forest{F}$.

\begin{Definition}[Tree product]\label{prodintrees}
	Let $\tree{t}_1 = (V_1, E_1)$ and $\tree{t}_2 = (V_2, E_2)$ be two trees with respective roots $r_1$ and $r_2$.
	The \emph{product} of $\tree{t}_1$ by $\tree{t}_2$ is the tree $\tree{t}_1 \treeProd \tree{t}_2$ whose vertex set is
	$V = \{(u,v) \in V_1 \times V_2 \mid \depth{u} = \depth{v}\}$,
	where $\depth{u}$ denotes the depth of $u$ (its distance to the root of its tree), and whose arc set is:
	\[
	E = \{((u,u'),(v,v')) \mid (u,v) \in E_1,\ (u',v') \in E_2\} \subseteq V^2.
	\]
\end{Definition}

\begin{figure}[t]
	\centering
	\includegraphics[scale=0.88]{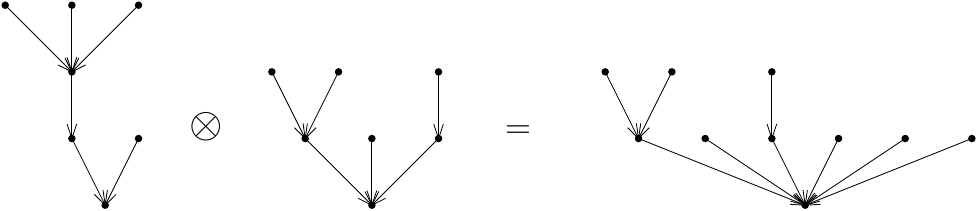}
	\caption{Example of a product of two trees.
		On the left, the two trees being multiplied; on the right, their product.
		Note that the depth of the result equals that of the shallowest factor.}\label{fig:tree_product}
\end{figure}

Figure~\ref{fig:tree_product} shows an example of the above definition. Equipped with these two operations, $\unroll{\FDDS}$ is a commutative semiring and $\unroll{\cdot}$ is a morphism between the semirings $(\FDDS, +, \times)$ and $(\unroll{\FDDS}, +, \treeProd)$.
Moreover, since both operations are also defined for finite trees, the set of forests equipped with these operations is also a semiring, and the set of finite forests of finite trees, denoted $\mathbb{F}_f$, is a pseudo-semiring, \ie, a semiring without a multiplicative identity.

The article~\cite{article_arbre} further defines a total order on trees that is compatible with the product for unroll trees: for all unroll trees $\tree{t}_1, \tree{t}_2$ and $\tree{t}_3$, we have $\tree{t}_1 \treeProd \tree{t}_2 \le \tree{t}_1 \treeProd \tree{t}_3$ if and only if $\tree{t}_2 \le \tree{t}_3$.

Using this order, the authors of~\cite{article_arbre} showed that unroll trees are \emph{cancellable}, \ie, for all unroll trees $\tree{t}_1, \tree{t}_2, \tree{t}_3$, we have $\tree{t}_1 \treeProd \tree{t}_2 = \tree{t}_1 \treeProd \tree{t}_3$ if and only if $\tree{t}_2 = \tree{t}_3$.
Moreover, it is shown that $\unroll{A} \treeProd \forest{X} = \unroll{B}$, with $A$ and $B$ dendrons, can be solved using the \emph{cut} of $\unroll{A}$ and $\unroll{B}$ at a certain depth.
This is of particular importance since~\cite{article_arbre} also proves that the equation $\tree{a} \treeProd \tree{x} = \tree{b}$ can be solved efficiently.

\begin{Definition}[Depth and cut]
	Let $\forest{F}$ be a forest and $u$ a vertex of $\forest{F}$.
	The \emph{depth} of $u$ in $\forest{F}$, denoted $\depth{u}$, is the distance from $u$ to the root of the tree containing it.
	Let $\tree{t}$ be a tree in $\forest{F}$.
	The \emph{depth} of $\tree{t}$, denoted $\depth{\tree{t}}$, is the maximum depth among the vertices of $\tree{t}$.
	Finally, the \emph{depth} of $\forest{F}$, denoted $\depth{ \forest{F} }$, is the maximum depth among its component trees.

	Let $n \ge 0$ be an integer.
	The \emph{cut} of $\forest{F}$ at depth $n$, denoted $\cut{\forest{F}}{n}$, is the forest obtained from $\forest{F}$ by removing all vertices of depth greater than $n$.
\end{Definition}

These results were extended in~\cite{kroot,automata2024ext,poly_inj,Antonio2026} and led to the proof of a property analogous to cancellability.

\begin{Lemma}[\cite{article_arbre, kroot,automata2024ext}]\label{lemme:cancelFiniForest}\label{lemma:casi_annulable}
	Given $\forest{A}$, $\forest{X}$, $\forest{Y} \in \mathbb{F}$, $\forest{A} \treeProd \forest{X} = \forest{A} \treeProd \forest{Y}$ implies $\cut{\forest{X}}{\depth{\forest{A}}} = \cut{\forest{Y}}{\depth{\forest{A}}}$.
\end{Lemma}

Another important result is that any polynomial over finite trees with a sufficiently deep non-constant coefficient is injective over a certain semiring.

\begin{prop}[\cite{poly_inj,Antonio2026}]\label{prop:injDesFinis}
	Let $P = \sum_{i=0}^{m} \forest{A}_i \treeProd \forest{X}^i$ with $d_{max} = \max_{i>0}(\depth{\forest{A}_i})$ and $\depth{\forest{A}_0} \le d_{max}$.
	Then $P$ is injective over the semiring of forests of finite trees of depth at most $d_{max}$.
\end{prop}

Exploiting the previous proposition, it is possible to prove that any polynomial over unrolls is injective.

\begin{Theorem}[\cite{poly_inj,Antonio2026}]\label{th:injPolyUnrolls}
	All univariate polynomials over unrolls are injective.
\end{Theorem}

Furthermore, it is shown that solving $P(\forest{X}) = \forest{B}$ can be done in polynomial time, where $P$ is a polynomial over finite forests.
Moreover, $Q(\forest{X}) = \unroll{B}$, with $Q = \sum_{i=0}^{m} \unroll{A_i} \treeProd \forest{X}^i$, can also be solved in polynomial time with respect to the sum of the sizes of $B$ and the $A_i$.
Additionally,~\cite{poly_inj,Antonio2026} provide a characterization of injective polynomials over FDDS.

\begin{Theorem}[\cite{poly_inj, Antonio2026}]\label{th:cara_poly_inj}
	Let $P = \sum_{i=0}^{m} A_i X^i$ be a polynomial over FDDS.
	Then, $P$ is injective if and only if at least one of its nonconstant coefficients, namely some $A_i$ with $i \neq 0$, contains a dendron.
\end{Theorem}

	\section{Algebraic properties of partial systems} \label{section:partel_introduction}

    A Partial Discrete Dynamical System (PDDS) is a pair $(C_X, f_X)$ where the transition function $f_X$ may be partial (some elements may have no image).
    Taken up to isomorphism, these objects form the set $\Partial$.
    Note that throughout this paper, partial dynamical systems are always considered up to isomorphism unless stated otherwise.
    Equipped with disjoint union as addition and direct product as multiplication, $\Partial$ forms a super-semiring of $(\FDDS, +, \times)$.
    The aim of this paper is to shed light on the algebraic structure of the semiring $(\Partial, +, \times)$.

	Let $A$ be a partial system.
	Remark that we can separate the connected components of $A$ into two categories (or submultisets): those having a cycle, denoted $\fullPart{A}$, and those that do not, denoted $\nilPart{A}$.
    Let us remark that  $\fullPart{\cdot}$ is a homomorphism between the semiring $(\Partial,+,\times)$ and $(\FDDS,+,\times)$. 
	The latter part is simply a finite forest.
	From there, we distinguish two remarkable subsets of partial systems: those where $\nilPart{X}$ is empty (the FDDS) and those where $\fullPart{X}$ is empty (the finite forests) (see Figure~\ref{fig:sys_partiel} for an example).

	\begin{figure}[t]
		\centering
		\includegraphics[page=1]{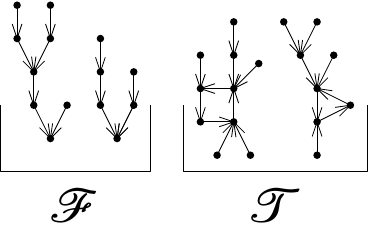}
		\caption{Example of a partial system $A$ highlighting $\nilPart{A}$ and $\fullPart{A}$.}\label{fig:sys_partiel}
	\end{figure}

	Remark that the set of finite forests is closed under addition and also under multiplication.
	Indeed, the direct product of two forests is a forest since each vertex is either derived from the product of two vertices with out-degree $1$ and thus has out-degree $1$, or from the product of a vertex with out-degree at most $1$ and a vertex with out-degree $0$ and thus has out-degree $0$.
	Thus, the set of finite forests, which, recall, is denoted $\mathbb{F}_f$, equipped with disjoint union as addition and direct product as multiplication, forms a \emph{pseudo-semiring}, \ie, a semiring without a multiplicative identity.
                 
	\begin{Lemma}
		The set $\mathbb{F}_f$ is an \emph{ideal} of $(\Partial,+,\times)$, \ie, $(\mathbb{F}_f,+)$ is a submonoid of $(\Partial,+)$ and for every element $A$ of $\Partial$ and every element $B$ of $\mathbb{F}_f$, we have that $A \times B$ and $B \times A$ are elements of $\mathbb{F}_f$.
	\end{Lemma}

	In this section, we seek to characterize the prime elements of $(\Partial,+,\times)$, if they exist, and to determine which polynomials induced by $(\Partial,+,\times)$ are injective.
	For the latter, we will extend the definition of unroll to the case of partial systems.
	Indeed, remark that the definition of unroll given for FDDS does not directly apply to partial systems since some connected components do not have a cycle, implying that the function $\unroll{X}$ is not total, which contradicts the definition of a morphism.

	\subsection{Prime elements of partial systems}\label{subsec:premier_partiel}

In this section, we show that the isolated point without a loop, simply denoted $\varrho$, is the unique prime element of $(\Partial,+,\times)$.
Recall that an element $p$ is prime in a semiring $(M,\diamond,\star)$ if it is nonzero, non-invertible and if for all $a,b$ in $M$, whenever there exists $k$ in $M$ such that $p\star k = a \star b$, there exists $k'$ such that $p \star k' = a$ or $p\star k' = b$.
In order to show that $\varrho$ satisfies this property, we characterize the form of the elements that $\varrho$ divides, and then conclude.

\begin{Lemma}\label{lemma:forme_div}
	Let $A$ and $K$ be partial systems.
	Then, $\varrho \times K = A$ if and only if $A$ is equal to $|K| \varrho$.
\end{Lemma}

\begin{proof}
	$(\Leftarrow)$ Suppose that $A$ is equal to $|K| \varrho$.
	Then, $\varrho$ divides $A$ directly.
	Indeed, if we consider $K$ as $|K|$ fixed points, then $K \times \varrho = |K| \varrho$ by definition.

	$(\Rightarrow)$ Suppose that $\varrho \times K = A$.
	Then, $A$ consists of $|K|$ vertices.
	Furthermore, since the out-degree of a vertex $(u,v)$ of a product is equal to the out-degree of $u$ times that of $v$, it follows that all vertices of $\varrho \times K$ have out-degree $0$.
	This implies that $A$ consists of $|K|$ copies of $\varrho$.
\end{proof}

\begin{prop}\label{prop:point_prime_in_partial}
	Let $A$ and $B$ be partial systems.
	Then, there exists a partial system $K$ such that $\varrho \times K = A \times B$ if and only if $A$ or $B$ consists only of copies of $\varrho$.
\end{prop}

\begin{proof}
	$(\Leftarrow)$ Suppose that $A$ or $B$ consists only of copies of $\varrho$.
	Without loss of generality, suppose that $A = |A| \varrho$.
	Then, we deduce from Lemma~\ref{lemma:forme_div} that $A \times B$ consists of $|A| \times |B|$ copies of $\varrho$.
	And from the same lemma, it follows that $\varrho$ divides $A \times B$ and $A$.

	$(\Rightarrow)$ Suppose that there exists a partial system $K$ such that $\varrho \times K = A \times B$.
	Then, by Lemma~\ref{lemma:forme_div}, it follows that $A \times B$ consists only of copies of $\varrho$.
	Now, since the out-degree of a vertex $(u,v)$ of a product is equal to the out-degree of $u$ times that of $v$, it follows that if $A$ and $B$ each contain a vertex of out-degree $1$, then so does $A \times B$.
	We deduce that $A$ or $B$ consists only of copies of $\varrho$.
\end{proof}

\begin{Corollary}\label{cor:varrho_prime}
	$\varrho$ is prime in $(\Partial,+,\times)$.
\end{Corollary}

\begin{proof}
	By definition, no element of $(\Partial,+,\times)$ other than the fixed point is invertible.
	Indeed, if we consider two partial systems $A$ and $B$, then the product $A \times B$ has $|A| \times |B|$ vertices.
	Thus, for $A \times B$ to be the fixed point, both $A$ and $B$ must consist of a single state.
	Moreover, this state must have an outgoing edge.
	Therefore, $A$ and $B$ are the fixed point.

	Thus, $\varrho$ is nonzero, non-invertible and, by Proposition~\ref{prop:point_prime_in_partial}, it satisfies the last condition for being prime.
\end{proof}

We now want to prove that $\varrho$ is the unique prime element of $(\Partial,+,\times)$.
For this, it suffices to prove that all partial systems with at least two vertices are not prime, since the multiplicative identity is by definition not prime.
This is the purpose of the following theorem.

\begin{Theorem}\label{th:partiel_prime}
	$\varrho$ is the unique prime element of $(\Partial,+,\times)$.
\end{Theorem}

\begin{proof}
	Let $A$ be a partial system with $n \ge 2$ vertices.
	We will show that $A$ is not prime.
	For this, it suffices to construct three partial systems $X,B,D$ such that $AX = BD$ but $A$ divides neither $B$ nor $D$.

	Let $n_1$ and $n_2$ be two divisors of $n$ such that $n = n_1 \cdot n_2$ and $n_1 < n_2$.
	Let $D$ be the partial system consisting of $n_2$ fixed points.
	We distinguish two cases depending on the value of $A$.

	The first case is that $A \neq D$.
	In this case, let $B = n_1 \varrho$.
	Observe that $A \times \varrho = B \times D$.
	Furthermore, $A$ divides neither $B$ nor $D$.
	Indeed, by construction $|A|> |B| >0$ and since the size of $AX$ is either $0$ (if $X = \emptyset$) or greater than or equal to the size of $A$ (if $X \neq \emptyset$), it follows that $A$ does not divide $B$.
	Moreover, remark that $|A| \ge |D|$.
	Thus, either $|A| > |D|$ and it follows that $A$ does not divide $D$.
	Or $|A| = |D|$.
	But in this case, for the equation $A \times X = D$ to have a solution, we need $|X| = 1$.
	However, $X$ cannot be the fixed point, since $A \neq D$.
	Therefore, the only possible value for $X$ is $\varrho$.
	But in this case, $A \times X$ contains no edge and is thus different from $D$.
	It follows that $A$ is not prime.

	The second possible case is that $A = D$.
	Then, we have $A \times \varrho = \varrho \times \cycle{n}$.
	But $A$ divides neither $B$ nor $\cycle{n}$ since a product with a nonempty FDDS does not reduce the number of connected components.
	This means that $A$ is not prime.
\end{proof}

However, remark that $\varrho$ is not irreducible.
Indeed, recall that an element is irreducible if it is neither zero, nor invertible, nor the product of two non-invertible elements.
Now, $\varrho$ is non-invertible but is the product of two non-invertible elements, since it is idempotent. 

	\subsection{Injectivity of polynomials over partial systems}\label{subsec:annu_partiel}

Let $A$ be a partial system.
In this section, we will prove that $A$ is cancelable in $(\Partial,+,\times)$ if and only if $A$ contains a dendron.
For this, recall that $A$, like every partial system, can be viewed as the sum of $\nilPart{A}$ and $\fullPart{A}$.
Thus, the product of $A$ with a partial system $B$ is:
\[
A \times B = \big(\nilPart{A} \times \nilPart{B}\big) + \big(\nilPart{A} \times \fullPart{B}\big) + \big(\fullPart{A} \times \nilPart{B}\big) + \big(\fullPart{A} \times \fullPart{B}\big).
\]

In order to prove our claim, we would like to describe exactly the product between two elements of $\Nilpotent$ as well as between an element of $\FDDS$ and an element of $\Nilpotent$.
Furthermore, since the elements of $\Nilpotent$ are forests, we would like this characterization to use $\treeProd$ instead of $\times$, \ie, to express $\times$ in terms of $\treeProd$.

For this purpose, we will use the function $\dt{\forest{X}}$ introduced in~\cite{article_arbre}.
Recall that if we consider a forest $\forest{f}$, then $\dt{\forest{f}}$ is the forest of trees rooted in the nodes at depth $0$ of $\forest{f}$. Moreover, we generalise this concept to a FDDS $\full{A}$ defining the $\dt{\fullPart{A}}$ as the forest of trees rooted in transient pre-images of cyclic nodes. 

We can generalize this by defining $\Di{i}{\forest{f}}$ as the multiset of trees of $\forest{f}$ rooted in a vertex of depth $i$.
Equivalently, the multiset $\Di{i}{\forest{f}}$ is obtained by applying the function $\dt{\forest{X}}$ to $\forest{f}$ a total of $i$ times.
We then have $\Di{i}{\Di{j}{\forest{f}}} = \Di{i+j}{\forest{f}} = \Di{j}{\Di{i}{\forest{f}}}$ for all integers $i,j \ge 0$ and all forests $\forest{f}$.
Furthermore, as shown by the following lemma, the function $\Di{i}{\forest{f}}$ is an endomorphism of $(\Nilpotent,+,\treeProd)$ for every integer $i\ge 0$.

\begin{Lemma} \label{lemma:Di_distribut_with_tree_product}
	Let $\forest{A}, \forest{B}$ be two forests and $i\ge0$ an integer.
	Then, $\Di{i}{\forest{A} \treeProd \forest{B}} = \Di{i}{\forest{A}} \treeProd \Di{i}{\forest{B}}$.
\end{Lemma}

\begin{proof}
	Let $\tree{t}$ be a tree of $\forest{A} \treeProd \forest{B}$.
	Then, there exist $\tree{t}_\forest{A}$ and $\tree{t}_\forest{B}$ in $\forest{A}$ and $\forest{B}$ respectively such that $\tree{t}_\forest{A} \treeProd \tree{t}_\forest{B} = \tree{t}$.
	Since the tree product is levelwise, each tree rooted in a vertex of depth $i$ in $\tree{t}$ is the product of two subtrees of $\tree{t}_\forest{A}$ and $\tree{t}_\forest{B}$ respectively, each having a root at depth $i$.
	It follows that $\Di{i}{\tree{t}_\forest{A}} \treeProd \Di{i}{\tree{t}_\forest{B}} = \Di{i}{\tree{t}}$.
	By summing over all trees, the lemma follows.
\end{proof}

\begin{Lemma} \label{lemma:caracterisation_nilpotent_product}
	Let $d\ge 0$ be an integer and let $\forest{A}, \forest{B}$ be two finite forests of depth at most $d$.
	Then, $\forest{A} \times \forest{B} = \forest{A} \treeProd \forest{B} + \sum_{i = 1}^{d} \big(\Di{i}{\forest{A}} \treeProd \forest{B} + \Di{i}{\forest{B}} \treeProd \forest{A}\big)$.
\end{Lemma}

An example of a direct product of finite forests is given in Figure~\ref{fig:prod_direct_tree}.

\begin{proof}
	By definition of the product of FDDS, we have that $\forest{A} \times \forest{B}$ is a forest with vertex set $V = \{(a,b) \mid a \in V(\forest{A}), b \in V(\forest{B})\}$ and edge set $E = \{((a,b) , (a',b')) \mid (a,a') \in E(\forest{A}), (b,b') \in E(\forest{B})\}$.
	Consider the forest $\forest{f} = \forest{A} \treeProd \forest{B} + \sum_{i = 1}^{d} \Di{i}{\forest{A}} \treeProd \forest{B} + \Di{i}{\forest{B}} \treeProd \forest{A}$ with vertex set $V'$ and edge set $E'$.
	We show that $\forest{A} \times \forest{B}$ and $\forest{f}$ are the same graph (exactly the same and not merely isomorphic).

	We first show that $V = V'$.
	For this, we prove the double inclusion ($V \subseteq V'$ and $V' \subseteq V$).
	Observe that all vertices of $V'$ are pairs $(a,b)$ with $a$ and $b$ vertices of $V(\forest{A})$ and $V(\forest{B})$ respectively.
	Therefore, since $V$ contains all such pairs, it follows that $V' \subseteq V$.

	Consider two trees $\tree{t}_\forest{A}$ and $\tree{t}_\forest{B}$ of $\forest{A}$ and $\forest{B}$ respectively, and vertices $a$ and $b$ of these trees at depths $d_{\tree{a}}$ and $d_{\tree{b}}$ respectively.
	Three cases are possible.
	Either $d_{\tree{a}} = d_{\tree{b}}$.
	Then, since $\forest{f}$ contains $\forest{A} \treeProd \forest{B}$, it follows that $V'$ contains the pair $(a,b)$.
	Or $d_{\tree{a}} < d_{\tree{b}}$.
	Then, since $\forest{f}$ contains $\tree{t}_\forest{A} \treeProd \Di{d_{\tree{b}}-d_{\tree{a}}}{\tree{t}_\forest{B}}$, it follows that $V'$ contains the pair $(a,b)$.
	Symmetrically, if $d_{\tree{a}} > d_{\tree{b}}$, then $V'$ contains the pair $(a,b)$.
	From this reasoning, we infer that each possible pair $(a,b)$ with $a$ and $b$ vertices of $\forest{A}$ and $\forest{B}$ respectively is an element of $V'$, implying that $V \subseteq V'$.

	It remains to show that $E = E'$.
	For this, we first show that $\forest{A} \times \forest{B}$ and $\forest{f}$ have the same roots.
	Directly, a vertex $(a,b)$ of $\forest{f}$ is a root if and only if $a$ or $b$ is a root.
	We prove that the same holds in $A \times B$.
	Suppose that $a$ or $b$ is a root of $\forest{A}$ or $\forest{B}$; then $(a,b)$ is a root of $A \times B$ since $a$ or $b$ has no successor.
	Conversely, if $a$ and $b$ are not roots of $\forest{A}$ and $\forest{B}$, then $a$ and $b$ have a successor in $\forest{A}$ and $\forest{B}$ respectively, denoted $a'$ and $b'$.
	Therefore, in $\forest{A} \times \forest{B}$, the vertex $(a,b)$ has $(a',b')$ as a successor.
	Thus, a vertex $(a,b)$ is a root of $\forest{A} \times \forest{B}$ if and only if $a$ or $b$ is a root of $\forest{A}$ or $\forest{B}$ respectively.
	We deduce that $\forest{f}$ and $\forest{A} \times \forest{B}$ have the same roots.

	From there, since $V = V'$, $\forest{F}$ and $\forest{A} \times \forest{B}$ have the same roots, and all other vertices of $\forest{F}$ and $\forest{A} \times \forest{B}$ have out-degree exactly $1$, it follows that $E$ and $E'$ have the same size.
	Therefore, to conclude that $E = E'$, it suffices to prove that $E' \subseteq E$.
	Let $((a,b),(a',b'))$ be an edge of $E'$.
	Three cases are possible:
	\begin{enumerate}
		\item The edge $((a,b),(a',b'))$ is an edge of $\forest{A} \treeProd \forest{B}$, hence $(a,a')$ is an edge of $\forest{A}$ and $(b,b')$ is an edge of $\forest{B}$.
		Therefore, $((a,b),(a',b'))$ is an element of $E$.
		\item There exists $i > 0$ such that $((a,b),(a',b'))$ is an edge of $\forest{A} \treeProd \Di{i}{\forest{B}}$.
		This implies that there exists an edge $(a,a')$ in $\forest{A}$ and an edge $(b,b')$ in $\Di{i}{\forest{B}}$ and hence in $\forest{B}$.
		Consequently, the edge $((a,b),(a',b'))$ is an element of $E$.
		\item There exists an integer $j > 0$ such that $((a,b),(a',b'))$ is an edge of $\Di{j}{\forest{A}} \treeProd \forest{B}$.
		This case is symmetric to the previous one, and therefore the edge $((a,b),(a',b'))$ is an element of $E$.
	\end{enumerate}
\end{proof}

\begin{figure}[t]
	\centering
	\includegraphics[page=12,width=\textwidth]{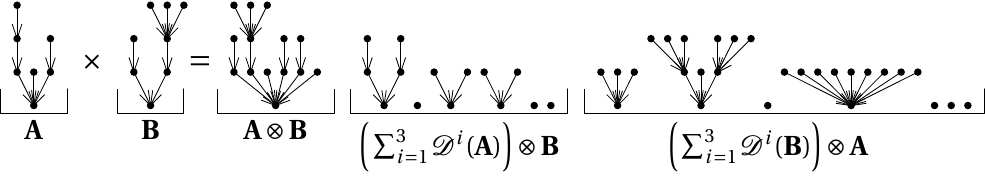}
	\caption{Example of a direct product of forests.}\label{fig:prod_direct_tree}
\end{figure}

It seems difficult to deduce properties about cancelable finite forests from the equality given in the statement of Lemma~\ref{lemma:caracterisation_nilpotent_product}.
Indeed, we cannot further factorize this equality, and it does not establish a clear link between the products $\times$ and $\treeProd$.
To address this, we introduce the following result.

\begin{Lemma}\label{lemma:reformule}
	Let $d\ge 0$ be an integer and let $\forest{A}$ and $\forest{B}$ be two finite forests of depth at most $d$.
	Then, $\sum_{i=0}^{d} \Di{i}{\forest{A} \times \forest{B}} = (\sum_{i=0}^{d} \Di{i}{\forest{A}}) \treeProd (\sum_{j=0}^{d} \Di{j}{\forest{B}})$.
\end{Lemma}

\begin{proof}
    We note that since $\forest{A}$ and $\forest{B}$ have depth at most $d$, we have $\Di{i}{\forest{A}} = \Di{i}{\forest{B}} = 0$ for all $i \ge d+1$. Let $S = \{ (i,j) : 0 \le i, j \le d, i+j \le 2d \}$. By Lemmas~ \ref{lemma:Di_distribut_with_tree_product},~\ref{lemma:caracterisation_nilpotent_product} and applying the function $\Di{i}{\forest{X}}$, we have
	\begin{align*}
	\sum_{i=0}^d \Di{i}{\forest{A} \times \forest{B}} &= \sum_{i=0}^{d} \Di{i}{\forest{A} \treeProd \forest{B}} + \sum_{i=0}^{d} \sum_{j=1}^{d} \Di{i}{\Di{j}{\forest{A}} \treeProd \forest{B}} + \sum_{i=0}^{d} \sum_{j=1}^{d} \Di{i}{\Di{j}{\forest{B}} \treeProd \forest{A}} \\
    &= \sum_{i = 0}^d \Di{i}{\forest{A}} \treeProd \Di{i}{\forest{B}} 
    + \sum_{i = 0}^d \sum_{j = 1}^d \Di{i+j}{\forest{A}} \treeProd \Di{i}{\forest{B}} 
    + \sum_{i = 0}^d \sum_{j = 1}^d \Di{i}{\forest{A}} \treeProd \Di{i+j}{\forest{B}} \\
    &= \sum_{i,j \in S: i = j} \Di{i}{\forest{A}} \treeProd \Di{j}{\forest{B}} 
    + \sum_{i,j \in S: i > j} \Di{i}{\forest{A}} \treeProd \Di{j}{\forest{B}} 
    + \sum_{i,j \in S: i < j} \Di{i}{\forest{A}} \treeProd \Di{j}{\forest{B}} \\
    &= \sum_{i,j \in S} \Di{i}{\forest{A}} \treeProd \Di{j}{\forest{B}} \\
    &= \sum_{i = 0}^d \Di{i}{\forest{A}} \treeProd \sum_{j = 0}^d \Di{j}{\forest{B}}.
	\end{align*}
\end{proof}

Remark that $\sum_{i=0}^{d} \Di{i}{\tree{t}}$ amounts to applying the unroll function to all vertices of $\tree{t}$ if $d$ is sufficiently large.
Building on this observation, we redefine the unroll of a system so that it also covers the partial case.
Figure~\ref{fig:unroll_partiel} gives two examples of the unroll of a partial system.

\begin{Definition}[Partial-system unroll]
	Let $A$ be a PDDS. The unroll of $A$, denoted $\unrollP{A}$ is the forest $\sum_{i=0}^{d} \Di{i}{A}$ where $d = \max(\depth{\nilPart{A}},\depth{\fullPart{A}})$ and $\Di{0}{A} = \nilPart{A} + \unroll{\fullPart{A}}$.
\end{Definition}
Intuitively, $\unrollP{A}$ is the forest obtained as the sum of $\unroll{\fullPart{A}}$ plus $\nilPart{A}$ plus the set of trees of $A$ rooted in vertices of depth at least $1$.
In the following, we will simply denote $\unroll{A}$ instead of $\unroll{\fullPart{A}}$ since the $\unroll{\cdot}$ operation is defined over FDDS.

\begin{figure}[t]
	\centering
	\includegraphics[page=2,scale=0.9]{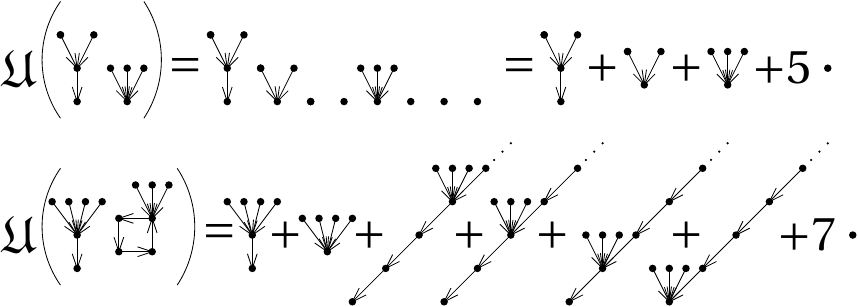}
	\caption{Example of the unroll of a forest and of a partial system with the new definition of unroll.}\label{fig:unroll_partiel}
\end{figure}

From the previous definition, we can rewrite Lemma~\ref{lemma:reformule} by using unroll notation, and we obtain the following result.

\begin{Lemma}\label{lemma:unroll_distributivity_over_product_nil}
	Let $\forest{A}, \forest{B}$ be finite forests.
	Then, $\unrollP{\forest{A} \times \forest{B}} = \unrollP{\forest{A}} \treeProd \unrollP{\forest{B}}$.
\end{Lemma}

In what follows, we will prove that $\unrollP{X}$ is a morphism between the semiring $(\Partial,+,\times)$ and $(\unrollP{\Partial},+,\treeProd)$ where $\unrollP{\Partial}$ denotes the set of unrolls obtainable from elements of $\Partial$.
But for now, we continue to focus on the pseudo-semiring $(\Nilpotent,+,\times)$ and obtain the following result.

\begin{prop}\label{proposition:unroll_is_isomorphism}
	The function $\unrollP{X}$ is an isomorphism between $(\Nilpotent, +, \times)$ and $(\unrollP{\Nilpotent}, +, \treeProd)$, where $\unrollP{\Nilpotent}$ denotes the set of unrolls obtainable from elements of $\Nilpotent$.
\end{prop}

\begin{proof}
	Since by definition $\unrollP{A + B} = \unrollP{A} + \unrollP{B}$ for all partial systems $A$ and $B$, it follows that Lemma~\ref{lemma:unroll_distributivity_over_product_nil} implies that the unroll is a morphism between $(\Nilpotent, +, \times)$ and $(\unrollP{\Nilpotent}, +, \treeProd)$.
	We still need to show that the function $\unrollP{X}$ is a bijection.
	Since, by definition of the codomain of $\unrollP{X}$, the partial-system unroll is surjective, it suffices to prove that it is also injective.

	Let $\forest{A}, \forest{B}$ be two finite forests such that $\forest{A} \neq \forest{B}$.
	Without loss of generality, we suppose that $\forest{A} = \tree{a}_1 + \cdots + \tree{a}_{m_\forest{A}}$ and $\forest{B} = \tree{b}_1 + \cdots + \tree{b}_{m_\forest{B}}$ and that these two forests are sorted according to the total order on trees.

	Since $\forest{A} \neq \forest{B}$, there exists an integer $i$ such that $\tree{a}_i \neq \tree{b}_i$.
	Two cases are then possible.
	Either the number of occurrences of $\tree{a}_i$ differs between $\unrollP{\forest{A}}$ and $\unrollP{\forest{B}}$, meaning that $\unrollP{\forest{A}} \neq \unrollP{\forest{B}}$.
	Or they have the same number of occurrences.
	Then, since the number of occurrences of $\tree{a}_i$ differs between $\forest{A}$ and $\forest{B}$, we deduce that the number of trees producing $\tree{a}_i$ via application of the function $\Di{j}{\forest{X}}$ differs between $\forest{A}$ and $\forest{B}$.
	Consequently, there exist two other trees $\tree{a}_{i'}$ and $\tree{b}_{i'}$ which are different and have depth strictly greater than that of $\tree{a}_i$.

	By applying this reasoning up to the maximum depth, we deduce that there exist two trees that do not have the same number of occurrences between $\unrollP{\forest{A}}$ and $\unrollP{\forest{B}}$, hence $\unrollP{\forest{A}} \neq \unrollP{\forest{B}}$.
	This means that $\unrollP{X}$ is injective and thus bijective.
\end{proof}

Since $\unrollP{X}$ is a bijective function between $\Nilpotent$ and $\unrollP{\Nilpotent}$, there exists an inverse function that allows us to reconstruct an element of $\Nilpotent$ from its unroll.
To this end, we also need a method for verifying whether a finite forest is an unroll.
In fact, we can do both simultaneously.
Indeed, if we consider a finite forest $\forest{A}$, we can compute the unroll of its tallest trees.
Then we can add these tallest trees to a solution (they cannot come from the unroll of a taller tree) and verify that their unroll is a submultiset of $\forest{A}$.
If this is the case, then we remove this unroll from $\forest{A}$ and reapply this procedure to the resulting system.
Otherwise, $\forest{A}$ is not an unroll and we return false.
The correctness proof of this procedure is straightforward and it runs in polynomial time with respect to the size of the tested forest.
This procedure is formalized in Algorithm~\ref{algo:inverse_unroll}.

\begin{algorithm}[t]
	\caption{\texttt{InverseUnroll}$(\forest{F})$}\label{algo:inverse_unroll}
	\begin{algorithmic}[1]
		\Require a finite forest $\forest{F}$
		\Ensure a finite forest $\forest{X}$ such that $\unrollP{\forest{X}} = \forest{F}$, or $\perp$
		\State $\forest{X} \gets \emptyset$,
		\State $\forest{R} \gets \forest{F}$,
		\While{$\forest{R} \neq \emptyset$}
		\State let $d_{max}$ be the maximum depth of a tree in $\forest{R}$,
		\State let $\forest{T}$ be the sub-multiset of trees of $\forest{R}$ having depth $d_{max}$,
		\State $\forest{X} \gets \forest{X} + \forest{T}$,
		\For{each tree $\tree{t} \in \forest{T}$}
		\If{$\unrollP{\tree{t}} \not\subseteq \forest{R}$}
		\State \Return $\perp$,
		\EndIf
		\State $\forest{R} \gets \forest{R} - \unrollP{\tree{t}}$,
		\EndFor
		\EndWhile
		\State \Return $\forest{X}$.
	\end{algorithmic}
\end{algorithm}

We now characterize which $\forest{X}$ and $\forest{Y}$ are possible such that $\forest{A} \times \forest{X} = \forest{A} \times \forest{Y}$.

\begin{prop} \label{theorem:caracterization_of_counter-example_of_cancelation}
	Let $\forest{A}$, $\forest{X}$ and $\forest{Y}$ be three finite forests.
	Then $\forest{A} \times \forest{X} = \forest{A} \times \forest{Y}$ if and only if the cuts of $\forest{X}$ and $\forest{Y}$ at depth $\depth{\forest{A}}$ are equal.
\end{prop}

\begin{proof}
	We proceed by equivalence.
	By Proposition~\ref{proposition:unroll_is_isomorphism}, we have $\forest{A} \times \forest{X} = \forest{A} \times \forest{Y}$ if and only if $\unrollP{A} \treeProd \unrollP{X} = \unrollP{A} \treeProd \unrollP{Y}$.
	From there, by Lemma~\ref{lemme:cancelFiniForest}, it follows that $\unrollP{A} \treeProd \unrollP{X} = \unrollP{A} \treeProd \unrollP{Y}$ if and only if $\cut{\unrollP{X}}{\depth{\forest{A}}} =\cut{\unrollP{Y}}{\depth{\forest{A}}}$.
\end{proof}

From Proposition~\ref{theorem:caracterization_of_counter-example_of_cancelation}, for every finite forest $\forest{A}$ there exist at least two distinct finite forests $\forest{X}$ and $\forest{Y}$ such that $\forest{A} \times \forest{X} = \forest{A} \times \forest{Y}$.
Indeed, for every finite forest $\forest{A}$, we can always construct a finite forest $\forest{X}$, consisting of a path $\tree{x}_1$ of depth $2 \depth{\forest{A}} + 2$ and a path $\tree{x}_2$ of depth $2 \depth{\forest{A}}$, and another finite forest $\forest{Y}$, consisting of two paths $\tree{y}_1,\tree{y}_2$ of depth $2 \depth{\forest{A}} + 1$.
We then have $\cut{\unrollP{X}}{\depth{\forest{A}}} = \cut{\unrollP{Y}}{\depth{\forest{A}}}$.
Indeed, the cut of $\unrollP{\tree{x}_1}$ at depth $\depth{\forest{A}}$ produces $\depth{\forest{A}} + 3$ copies of the path of depth $\depth{\forest{A}}$ plus one copy of each path having depth between $\depth{\forest{A}}-1$ and $0$.
The cut of $\unrollP{\tree{x}_2}$ at depth $\depth{\forest{A}}$ produces $\depth{\forest{A}} + 1$ copies of the path of depth $\depth{\forest{A}}$ plus one copy of each path having depth between $\depth{\forest{A}}-1$ and $0$.
Finally, the cuts of $\unrollP{\tree{y}_1}$ and $\unrollP{\tree{y}_2}$ at depth $\depth{\forest{A}}$ produce $\depth{\forest{A}} + 2$ copies of the path of depth $\depth{\forest{A}}$ plus one copy of each path having depth between $\depth{\forest{A}}-1$ and $0$.
This concludes the proof of the following theorem.

\begin{Theorem} \label{theorem:none_nilpotent_are_cancelable}
	No finite forest is cancelable in the pseudo-semiring $(\Nilpotent,+,\times)$.
\end{Theorem}

Thanks to Theorem~\ref{theorem:none_nilpotent_are_cancelable}, we know that if a partial system $A$ is cancelable, then $\fullPart{A}$ is nonempty.
We still need to show that $\fullPart{A}$ must contain a dendron for $A$ to be cancelable.
For this, as stated previously, we characterize the behavior of $\forest{A} \times B$ with $B$ an FDDS.

\begin{Lemma}\label{lemma:caracterisation_nilpotent_full_product}
	Let $\forest{A}$ be a finite forest and $B$ an FDDS.
	Then $\forest{A} \times B$ is equal to $\forest{A} \treeProd \unrollP{B}$.
\end{Lemma}

\begin{proof}
	In this proof $\unrollP{B}$ is not taken up to isomorphism.
	In other words, each vertex has a label of the form $(b,k)$ with $b$ a vertex of $B$ and $k\ge0$ an integer.
	Let $V$ and $E$ be the vertex and edge sets of $\forest{A} \times B$ and let $V'$ and $E'$ be those of $\unrollP{B} \treeProd \forest{A}$.
	In order to show that $\forest{A} \times B$ is isomorphic to $\unrollP{B} \treeProd \forest{A}$, we start by proving that these two graphs have the same number of vertices.

	By definition, $V$ contains the Cartesian product of the vertices of $\forest{A}$ with those of $B$.
	Therefore $|V| = |\forest{A}| \cdot |B|$.
	We thus need to show that $V'$ has the same size.
	Consider $a$ a vertex of $\forest{A}$ of depth $d$.
	By definition of the tree product, we have that $V'$ contains a vertex $(a,(b,k))$ for every vertex $(b,k)$ of $\unrollP{B}$ whose depth is $d$.
	Since the depth of $(b,k)$ in $\unrollP{B}$ is $k$, we deduce that $V'$ contains all vertices of the form $(a,(b,d))$.
	Furthermore, remark that $V'$ does not contain any vertex of the form $(a,(b,k))$ if $k \neq d$.
	Now, by definition of the vertices of $\unrollP{B}$, we deduce that all vertices of the form $(b,d)$ with $b$ a vertex of $B$ belong to $\unrollP{B}$.
	Therefore, a vertex $(a,(b,k))$ belongs to $V'$ if and only if $a$ and $b$ are respectively vertices of $\forest{A}$ and $B$ and $k$ is the depth of $a$ in $\forest{A}$.
	Consequently, we can define two surjective functions $f: V' \to V$ and $g : V \to V'$ such that $f\bigg(\big(a,(b,k)\big)\bigg) = (a,b)$ and $g\big((a,b)\big) = \big(a,(b,k)\big)$ with $k$ the depth of $a$ in $\forest{A}$.
	We thus deduce that there exists a bijection between $V$ and $V'$, implying that these two sets have the same size.
	Furthermore, observe that $g$ is the inverse function of $f$.

	We show that $g$ is an isomorphism between $\forest{A} \times B$ and $\unrollP{B} \treeProd \forest{A}$.
	Consider an edge $\bigg(\big(a,(b,k)\big), \big(a',(b',k-1)\big)\bigg)$ of $E'$.
	By definition of the tree product, we deduce that the edge $(a,a')$ belongs to $\forest{A}$ and that the edge $\big((b,k), (b',k-1)\big)$ belongs to $\unrollP{B}$.
	Since the definition of partial-system unrolls implies that the edge $(b,b')$ belongs to $B$, we infer that the edge $\big((a,b),(a',b')\big)$ belongs to $E$.
	It follows that if $\bigg(\big(a,(b,k)\big), \big(a',(b',k-1)\big)\bigg)$ is an edge of $E'$ then $\biggl(g\bigg(\big(a,(b,k)\big)\bigg),g\bigg(\big(a',(b',k-1)\big)\bigg)\biggr)$ is an edge of $E$.
	By a similar reasoning, we prove that if $\big((a,b),(a',b')\big)$ is an edge of $E$, then $\bigg(f\big((a,b)\big),f\big((a',b')\big)\bigg)$ is an edge of $E'$.
	We conclude that the two graphs are isomorphic.
\end{proof}

\begin{Corollary}\label{corollary:caracterisation_nilpotent_permutation_product}
	Let $\forest{A}$ be a finite forest and $B$ a permutation.
	Then $\forest{A} \times B$ is equal to $|B| \forest{A}$.
\end{Corollary}

\begin{proof}
	By Lemma~\ref{lemma:caracterisation_nilpotent_full_product}, we have that $\forest{A} \times B$ is equal to $\unrollP{B} \treeProd \forest{A}$.
	Since $B$ is a permutation, its partial-system unroll consists of $|B|$ infinite simple paths.
	In other words, the partial-system unroll of $B$ is equal to $|B|$ times the identity of $(\mathbb{F},+,\treeProd)$.
\end{proof}

\begin{prop}\label{proposition:cond_necessary_for_cancellability}
	Let $A$ be a partial system.
	If $A$ is cancelable, then $A$ contains a dendron.
\end{prop}

\begin{proof}
	We reason by contraposition.
	Suppose that $A$ does not contain a dendron.
	Then, either $\fullPart{A}$ is empty and therefore $A$ is not cancelable by Theorem~\ref{theorem:none_nilpotent_are_cancelable}.
	Or $\fullPart{A}$ is nonempty.
    In this case, \cite{article_arbre} introduces two distinct permutations $\mathsf{X}$ and $\mathsf{Y}$ an such that $\fullPart{A} \times \mathsf{X} = \fullPart{A} \times \mathsf{Y}$. 
    And since $\mathsf{X}$ and $\mathsf{Y}$ are permutations we have $\nilPart{A} \times \mathsf{X} = \nilPart{A} \times \mathsf{Y}$ by Corollary~\ref{corollary:caracterisation_nilpotent_permutation_product}.
    This implies that $A \times \mathsf{X} = A  \times \mathsf{Y}$ with $\mathsf{X} \neq \mathsf{Y}$.
\end{proof}

We now show that if a partial system $A$ contains a dendron, then it is cancelable.
Let $X$ and $Y$ be two partial systems.
As stated previously, $A \times X = A\times Y$ implies:
	\begin{equation}\label{eq:grande_eg}
		\fullPart{X} \times \bigg(\fullPart{A} + \nilPart{A}\bigg) + \nilPart{X} \times \bigg(\fullPart{A} + \nilPart{A}\bigg) = \fullPart{Y} \times \bigg(\fullPart{A} + \nilPart{A}\bigg) + \nilPart{Y} \times \bigg(\fullPart{A} + \nilPart{A}\bigg).
	\end{equation}

From this equality, it follows that $\fullPart{A} \times \fullPart{X} = \fullPart{A} \times \fullPart{Y}$.
Now, by Theorem~\ref{th:cara_poly_inj}, the fact that $A$ contains a dendron implies that $\fullPart{X} = \fullPart{Y}$.
We deduce from~\eqref{eq:grande_eg} that
\[
\nilPart{X} \times (\fullPart{A} + \nilPart{A}) = \nilPart{Y} \times (\fullPart{A} + \nilPart{A}).
\]
In order to show that this equality implies $\nilPart{X} = \nilPart{Y}$, we prove that partial systems containing cycles are cancelable with respect to finite forests.
That is, if $A \times \forest{X} = A \times \forest{Y}$ with $\fullPart{A} \neq \emptyset$, then $\forest{X} = \forest{Y}$.
The proof relies on the total order on trees that is compatible with the product $\treeProd$.
In order to use this order simply, we introduce the following lemma.

\begin{Lemma}\label{lemma:unroll_distributivity_over_product_partial}
	The function $\unrollP{X}$ is a morphism between $(\Partial,+,\times)$ and $(\unrollP{\Partial},+,\treeProd)$.
\end{Lemma}

\begin{proof}
	By definition, $\unrollP{X}$ distributes over addition.
	Furthermore, still by definition, the partial-system unroll of the empty system is the empty forest and that of the fixed point is the infinite simple path.
	Therefore, the images of the additive and multiplicative identities of $(\Partial,+,\times)$ are respectively the additive and multiplicative identities of $(\unrollP{\Partial},+,\treeProd)$.
	It remains to show that $\unrollP{X}$ distributes over the product.

	Let $A$ and $B$ be two partial systems.
	In this proof, the partial-system unrolls as well as $A$ and $B$ are not taken up to isomorphism.
	We consider $\unrollP{A \times B}$ with vertex set $V$ and edge set $E$.
	Similarly, we consider $\unrollP{A} \treeProd \unrollP{B}$ with vertex set $V'$ and edge set $E'$.
	We show that $\unrollP{A \times B}$ is equal to $\unrollP{A} \treeProd \unrollP{B}$.
	For this, we define a bijective function between $V$ and $V'$ and prove that it is an isomorphism between $(V,E)$ and $(V',E')$.

	By definition, we have that $V$ is the set of vertices $\big((a,b),k\big)$ with $a$ and $b$ vertices of $A$ and $B$ and $k\ge 0$ an integer.
	Consider the bijective function $f$ that maps each vertex $\big((a,b),k\big)$ to $\big((a,k),(b,k)\big)$.
	Let $f(V)$ denote the set of images of the vertices of $V$ under $f$.
	We directly have that $f$ is a bijection between $V$ and $f(V)$.
	It remains to show that $f(V) = V'$.
	For this, recall that $\unrollP{A}$ and $\unrollP{B}$ have as vertex sets respectively all possible pairs $(a,k)$ and $(b,k)$ with $a$ and $b$ vertices of $A$ and $B$ respectively and $k\ge 0$ an integer.
	Furthermore, the vertices $(a,k)$ and $(b,k)$ have depth $k$ in $\unrollP{A}$ and $\unrollP{B}$ respectively.
	This implies that $V' = \{\big((a,k),(b,k)\big) \mid a \in A, b \in B, k\in \N\}$.
	Therefore, $f(V) = V'$.

	We now show that $f$ is an isomorphism between $\unrollP{A \times B}$ and $\unrollP{A} \treeProd \unrollP{B}$.
	Consider $e= \bigg(\big((a,b),k\big), \big((a',b'),k-1\big)\bigg)$.
	Then, by definition of partial-system unrolls, $e$ is an edge of $E$ if and only if $\big((a,b),(a',b')\big)$ is an edge of $A \times B$ and $k\ge1$.
	And by definition of the direct product, $\big((a,b),(a',b')\big)$ is an edge of $A \times B$ if and only if $(a,a')$ is an edge of $A$ and $(b,b')$ is an edge of $B$.
	Again, by definition of partial-system unrolls, $\big((a,k),(a',k-1)\big)$ is an edge of $\unrollP{A}$ if and only if $(a,a')$ is an edge of $A$ and $k\ge1$.
	Similarly, $\big((b,k),(b',k-1)\big)$ is an edge of $\unrollP{B}$ if and only if $(b,b')$ is an edge of $B$ and $k\ge1$.
	Therefore, by definition of the tree product, $\bigg(\big((a,k),(b,k)\big),\big((a',k-1),(b',k-1)\big)\bigg)$ is an edge of $E'$ if and only if $(a,a')$ is an edge of $A$, $(b,b')$ is an edge of $B$ and $k\ge1$.
	We conclude that $\bigg(\big((a,b),k\big), \big((a',b'),k-1\big)\bigg)$ is an edge of $E$ if and only if $\biggl(f\bigg(\big((a,b),k\big)\bigg), f\bigg(\big((a',b'),k-1\big)\bigg)\biggr)$ belongs to $E'$.
	The claim follows.
\end{proof}




\begin{Lemma}\label{lemma:Partial_withg_full_cancelable_over_nil}
	Let $A$ be a partial system such that $\fullPart{A} \neq \emptyset$.
	Let $\forest{X}$ and $\forest{Y}$ be two finite forests.
	If $A \times \forest{X} = A \times \forest{Y}$ then $\forest{X} = \forest{Y}$.
\end{Lemma}

\begin{proof}
	Suppose that $A \times \forest{X} = A \times \forest{Y}$.
	Then, by Lemma~\ref{lemma:unroll_distributivity_over_product_partial}, we obtain $\unrollP{A} \treeProd \unrollP{X} = \forest{B} = \unrollP{A} \treeProd \unrollP{Y}$.

	Consider $\tree{t}$ the smallest tree, according to the order compatible with the product, of $\forest{B}$ with maximum depth.

	We then have $\tree{a} \tree{x} = \tree{t} = \tree{a} \tree{y}$ with $\tree{a}$ the smallest tree of $\unrollP{A}$ of depth at least $\depth{\tree{t}}$ and $\tree{x}$ and $\tree{y}$ the smallest trees of $\unrollP{\forest{x}}$ and $\unrollP{\forest{y}}$ respectively of depth $\depth{\tree{t}}$.
	Indeed, we know that $\depth{\tree{x}} = \depth{\tree{t}} = \depth{\tree{y}}$ since $\unrollP{A}$ contains infinite trees.
	Furthermore, Lemma~\ref{lemma:casi_annulable} implies that $\cutT{x}{\depth{\tree{a}}} = \cutT{y}{\depth{\tree{a}}}$.
	Since the depth of $\tree{a}$ is greater than or equal to that of $\tree{x}$ and $\tree{y}$, $\cutT{x}{\depth{\tree{a}}} = \tree{x}$ and $\cutT{y}{\depth{\tree{a}}} = \tree{y}$.
	It follows that $\tree{x} = \tree{y}$ and consequently, $\unrollP{A} \treeProd (\unrollP{\forest{x}} - \tree{x}) = \unrollP{A} \treeProd (\unrollP{\forest{y}} - \tree{y})$.

	By applying this reasoning recursively, we deduce that $\unrollP{\forest{x}} = \unrollP{\forest{y}}$.
	From Proposition~\ref{proposition:unroll_is_isomorphism}, we conclude that $\forest{x} = \forest{y}$.
\end{proof}

Thus, returning to our equation
\[
\nilPart{X} \times (\fullPart{A} + \nilPart{A}) = \nilPart{Y} \times (\fullPart{A} + \nilPart{A}),
\]
Lemma~\ref{lemma:Partial_withg_full_cancelable_over_nil} implies that $\nilPart{X} = \nilPart{Y}$ since $\fullPart{A}$ contains a dendron (and is thus nonempty).
We conclude that if $A$ contains a dendron, then $\fullPart{X} = \fullPart{Y}$ and $\nilPart{X} = \nilPart{Y}$ for all partial systems $X$ and $Y$ such that $A \times X = A \times Y$.
This implies that $A$ is cancelable.
Combining this result with Proposition~\ref{proposition:cond_necessary_for_cancellability}, we conclude the proof of the following theorem.

\begin{Theorem}\label{theorem:caracterisation_cancellable_in_partial}
	Let $A$ be a partial system.
	Then $A$ is cancelable if and only if $A$ contains a dendron.
\end{Theorem}

We can extend this theorem to characterize injective polynomials over partial systems.

\begin{Theorem}\label{th:cara_poly_inj_partiel}
	Let $P = \sum_{i=0}^{m} A_i X^i$ be a polynomial over partial systems.
	Then, $P$ is injective if and only if there exists $i>0$ such that $A_i$ contains a dendron.
\end{Theorem}

\begin{proof}
	$(\Leftarrow)$ Suppose that there exists $j>0$ such that $A_j$ contains a dendron.
	Let $X$ and $Y$ be two partial systems such that $P(X) = P(Y)$.
	Since $\fullPart{P(X)} = \sum_{i=0}^{m} \fullPart{A_i} \fullPart{X}^i$, Theorem~\ref{th:cara_poly_inj} implies that $\fullPart{X} = \fullPart{Y}$.
	Furthermore, we have $\unrollP{P(X)} = \unrollP{P(Y)}$.
	Thus, since $\unrollP{P(X)} = \sum_{i=0}^{m} \unrollP{A_i} \unrollP{X}^i$ and since $\unrollP{A_j}$ contains an infinite tree, Proposition~\ref{prop:injDesFinis} implies that $\unrollP{X} = \unrollP{Y}$.
	Therefore $\unrollP{X} - \unrollP{\fullPart{X}} = \unrollP{Y} - \fullPart{Y}$ and Proposition~\ref{proposition:unroll_is_isomorphism} allows us to conclude that $\nilPart{X} = \nilPart{Y}$ and also that $X = Y$.

	$(\Rightarrow)$ Suppose that $A_i$ does not contain a dendron for any integer $i>0$.
	Then two cases are possible.
	Either $\fullPart{A_i} = \emptyset$ for every integer $i>0$.
	Consider $X = \cycle{1} + \cycle{1}$ and $Y = \cycle{2}$.
	Then $|X|^k = |Y|^k$ for every integer $k\ge 0$.
	Therefore, since all $A_i$ with $i>0$ are finite forests, Corollary~\ref{corollary:caracterisation_nilpotent_permutation_product} implies that $P(X) = P(Y)$.
	Hence $P$ is not injective.

	The other case is that $\fullPart{A_i} \neq \emptyset$ for at least one $i>0$.
	Then, using the permutations $X$ and $Y$ from section~3.3 of~\cite{Antonio2026}, it follows that $\nilPart{P(X)} = \nilPart{P(Y)}$ by Corollary~\ref{corollary:caracterisation_nilpotent_permutation_product} since $X$ and $Y$ are permutations of the same size, and that $\fullPart{P(X)} = \fullPart{P(Y)}$ by Theorem~17 of~\cite{Antonio2026}.
	We conclude that $P(X) = P(Y)$, hence $P$ is not injective.
\end{proof}

	\section{Solving $P(X) = B$ for partial systems}\label{sec:eqPartiel}

	In this section, we will prove that solving equations of the form $P(X) = B$ in the case of partial systems is no harder than solving $P(X) = B$ in the case of FDDS.
	More precisely, we will prove that if we have a polynomial-time algorithm for solving $P(X) = B$ with $P$ a polynomial over FDDS and $B$ an FDDS, then we also have one for solving $P'(X) = B'$ in polynomial time with $P'$ a polynomial over partial systems constructed by adding finite forests to the coefficients of $P$ and $B'$ a partial system such that $\fullPart{B'} = B$.

	Let $P = \sum_{i=0}^{m} A_i X^i$ be a polynomial over partial systems and let $B$ be a partial system.
	Directly, we observe that solving $P(X) = B$ is equivalent to solving $P(X) - A_0 = B - A_0$ if $A_0$ is a submultiset of $B$, otherwise no solution exists.
    
	Therefore, in what follows, we only consider polynomials without a constant term.
	We will begin by showing that if at least one coefficient of $P$ contains a cycle and there exists a solution $X$ to $P(X) = B$, then we can find $\nilPart{X}$ in polynomial time (Section~\ref{subsection:coef_full}).
	Then we will prove that if all coefficients of $P$ are finite forests, then we solve $P(X) = B$ in polynomial time (Section~\ref{section:division_in_Nilpotent}).

		\subsection{Solving $P(X) = B$ when some coefficient has a nonzero total part}\label{subsection:coef_full}

Let $P = \sum_{i=1}^{m} A_i X^i$ be a polynomial over partial systems without constant term such that one of the coefficients, say $A_i$, contains a cycle.
Let $B$ be a partial system such that the equation $P(X) = B$ has a solution.
In order to show that we can recover the finite forest part of the solution, we start by extending Lemma~\ref{lemma:Partial_withg_full_cancelable_over_nil} to the case of polynomials.
In other words, we show that all solutions of $P(X) = B$ have the same finite forest part.
After that, we only need to prove that it is possible to construct this forest in polynomial time.

\begin{Lemma}\label{lemma:inj_des_finis}
	Let $P = \sum_{i=1}^{m} A_i X^i$ be a polynomial over partial systems without constant term such that there exists $i$ with $A_i$ containing a cycle.
	Let $X,Y$  be partial systems.
	If $P(X) = P(Y)$ then $\nilPart{X} = \nilPart{Y}$.
\end{Lemma}

\begin{proof}
	Suppose that $P(X) = P(Y)$.
	Then, by Lemma~\ref{lemma:unroll_distributivity_over_product_partial}
	\[
	\sum_{i=1}^{m} \unrollP{A_i} \treeProd \unrollP{X}^i = \sum_{i=1}^{m} \unrollP{A_i} \treeProd \unrollP{Y}^i.
	\]
	Thus, by Proposition~\ref{prop:injDesFinis}, it follows that $\unrollP{X} = \unrollP{Y}$.
	Indeed, since one of the $\unrollP{A_i}$ contains a tree of infinite depth, Proposition~\ref{prop:injDesFinis} implies that $\sum_{i=1}^{m} \unrollP{A_i} \treeProd \unrollP{X}^i$ is injective over $(\mathbb{F},+,\treeProd)$.
	From there, by Proposition~\ref{proposition:unroll_is_isomorphism}, in order to show that $\nilPart{X} = \nilPart{Y}$, it suffices to prove that $\unrollP{\nilPart{X}} = \unrollP{\nilPart{Y}}$.
	Or equivalently, that $\unrollP{\fullPart{X}} = \unrollP{\fullPart{Y}}$.
	Remark that if $\unroll{\fullPart{X}} = \unroll{\fullPart{Y}}$ then $\unrollP{\fullPart{X}} = \unrollP{\fullPart{Y}}$.
	Indeed, we have that $\unrollP{\fullPart{A}}$ is the sum of $\unroll{\fullPart{A}}$ with the partial-system unroll of each finite tree rooted in the predecessor vertices of the root of a tree in $\unroll{\fullPart{A}}$.
	Thus, in order to prove our claim, it suffices to show that $\unroll{\fullPart{X}} = \unroll{\fullPart{Y}}$.

	For this, remark that by hypothesis $\sum_{i=1}^{m} \unrollP{A_i} \treeProd \unrollP{X}^i$ and $\sum_{i=1}^{m} \unrollP{A_i} \treeProd \unrollP{Y}^i$ have the same infinite trees.
	Since, by definition of the product, these are derived from the product of infinite trees of the $\unrollP{A_i}$ and $\unrollP{X}^i$ or $\unrollP{Y}^i$, and since every infinite tree of $\unrollP{A}$ comes from $\unroll{\fullPart{A}}$, we deduce that
	\[
	\sum_{i=1}^{m} \unroll{\fullPart{A_i}} \treeProd \unroll{\fullPart{X}}^i = \sum_{i=1}^{m} \unroll{\fullPart{A_i}} \treeProd \unroll{\fullPart{Y}}^i.
	\]
	Thus, by Theorem~\ref{th:injPolyUnrolls}, we conclude that $\unroll{\fullPart{X}} = \unroll{\fullPart{Y}}$.
\end{proof}

In order to show that we can find the finite forest part of a solution to the equation $P(X) = B$, we prove that we can efficiently solve the equation $\sum_{i=1}^{m} \unrollP{A_i} \treeProd \unrollP{X}^i = \unrollP{B}$.
For this purpose, we would like to use the algorithm introduced in~\cite{poly_inj,Antonio2026} for solving equations over finite forests in polynomial time.
But for this, we need to extend Theorem 47 of~\cite{Antonio2026} to the case of partial-system unrolls.

\begin{prop}\label{prop:polyUnrollP2PolyForestFini}
	Let $P = \sum_{i=1}^{m} A_i X^i$ be a polynomial over partial-system, let $B$ be a partial-system and $\alpha$ the number of trees in $\unroll{B}$.
	Let $n \ge 2 \cdot \alpha^2 + \depth{B} + 1$ be an integer.
	Then, there exists a finite forest $\forest{y}$ of depth $n$ such that $\sum_{i=0}^{m}\cut{\unrollP{A_i}}{n} \forest{y}^i = \cut{\unrollP{B}}{n}$ if and only if there exists an partial-system $X$ such that $\unrollP{P(X)} = \unrollP{B}$.
	Furthermore, if such an $X$ exists, then $\cut{\unrollP{X}}{n} = \forest{y}$.
\end{prop}

\begin{proof}
    $(\Rightarrow)$ Suppose that there exists a finite forest $\forest{y}$ of depth $n$, such that $\sum_{i=0}^{m}\cut{\unrollP{A_i}}{n} \forest{y}^i = \cut{\unrollP{B}}{n}$.
	Let $\forest{Z}$ be the submultiset of $\forest{y}$ containing only the trees of depth $n$.
	Then, we have $\sum_{i=0}^{m}\cut{\unroll{A_i}}{n} \forest{z}^i = \cut{\unroll{B}}{n}$.
	Indeed, no finite tree of the $\unrollP{A_i}$ and $\unrollP{B}$ has depth $n$.
	Thus, by Theorem 47 of~\cite{Antonio2026}, it follows that there exists an FDDS $X$ such that $\cutUn{X} = \forest{z}$ and hence $\cut{\unrollP{X}}{n} = \forest{y}$.

	$(\Leftarrow)$ Suppose that there exists an FDDS $X$ such that $\unrollP{P(X)} = \unrollP{B}$.
	Then, in particular $\unroll{P(X)} = \unroll{B}$ since the set of infinite trees of $\unrollP{P(X)}$ and of $\unrollP{B}$ is contained in $\unroll{P(X)}$ and $\unroll{B}$ respectively.

	The direction follows by Theorem 47 of~\cite{Antonio2026}.
\end{proof}

Building on Proposition~\ref{prop:polyUnrollP2PolyForestFini}, we derive Algorithm~\ref{algo:find_nilpart}, able of finding the finite forest part $\nilPart{X}$ of a solution $X$ to $P(X) = B$, if a solution exists.
For this, it suffices to compute each $\cut{\unrollP{A_i}}{n}$ and $\cut{\unrollP{B}}{n}$ with $n$ greater than or equal to $2 \cdot |B|^2 + \max(\depth{\nilPart{B}},\depth{\fullPart{B}})$.
Then, we solve the resulting equation over finite trees.
Let $\forest{X}$ denote the forest thus obtained.
Since the unroll and the cut operations are morphisms, if there exists a solution to $P(X) = B$, then we know that $\forest{X}$ is a well-defined forest.

After that, we compute each $\cut{\unrollP{\fullPart{A_i}}}{n}$ and $\cut{\unrollP{\fullPart{B}}}{n}$ and solve the associated equation.
Let $\forest{T}$ denote the forest thus obtained.
Observe that if there exists $X$ such that $P(X) = B$ then $\sum_{i=1}^{m} \fullPart{A_i} \fullPart{X}^i = \fullPart{B}$.
Consequently, we have $\sum_{i=1}^{m} \unrollP{\fullPart{A_i}} \unrollP{\fullPart{X}}^i = \unrollP{\fullPart{B}}$.
And by Proposition~\ref{prop:polyUnrollP2PolyForestFini}, this equation has a solution if and only if the equation $\sum_{i=1}^{m} \cut{\unrollP{\fullPart{A_i}}}{n} \cut{\unrollP{\fullPart{X}}}{n}^i = \cut{\unrollP{\fullPart{B}}}{n}$ also has one.
This means that $\forest{T}$ is well-defined.

Now, it suffices to remove $\forest{T}$ from $\forest{X}$ to obtain $\forest{F}$.
Remark that if there indeed exists a solution $X$ to $P(X) = B$, then Proposition~\ref{prop:injDesFinis} implies $\cut{\unrollP{X}}{n} = \forest{X}$.
But also that $\cut{\unrollP{\fullPart{X}}}{n} = \forest{T}$, since $n$ is large enough that all finite trees of $\unrollP{X}$ have depth strictly less than $n$.
Therefore:
\[
\cut{\unrollP{\nilPart{X}}}{n} = \unrollP{\nilPart{X}} = \forest{X} - \forest{T} = \forest{F}.
\]

Finally, it suffices to verify that $\forest{F}$ is indeed the partial-system unroll of a finite forest and to reconstruct said forest using Algorithm~\ref{algo:inverse_unroll}.
The correctness follows from Proposition~\ref{proposition:unroll_is_isomorphism}.

This procedure is formalized in Algorithm~\ref{algo:find_nilpart}.

\begin{algorithm}[t]
	\caption{\texttt{FindNilPart}$(P, B)$}\label{algo:find_nilpart}
	\begin{algorithmic}[1]
		\Require a polynomial $P = \sum_{i=1}^{m} A_i X^i$ over partial systems (with some $A_i$ containing a cycle), a partial system $B$
		\Ensure a finite forest $\nilPart{X}$ where $X$ is a solution to $P(X) = B$, or $\perp$
		\State let $n \gets 2 \cdot |B|^2 + \max(\depth{\nilPart{B}},\depth{\fullPart{B}})$,
		\State \textbf{// Solve over all partial-system unrolls}
		\For{$i = 1, \ldots, m$}
		\State compute $\cut{\unrollP{A_i}}{n}$,
		\EndFor
		\State compute $\cut{\unrollP{B}}{n}$,
		\State solve $\sum_{i=1}^{m} \cut{\unrollP{A_i}}{n} \cdot \forest{Y}^i = \cut{\unrollP{B}}{n}$ using algorithm~ of Section~5 of~\cite{Antonio2026},
		\If{no solution exists}
		\State \Return $\perp$,
		\EndIf
		\State let $\forest{X}$ be the solution,
		\State \textbf{// Solve over full-part unrolls only}
		\For{$i = 1, \ldots, m$}
		\State compute $\cut{\unrollP{\fullPart{A_i}}}{n}$,
		\EndFor
		\State compute $\cut{\unrollP{\fullPart{B}}}{n}$,
		\State solve $\sum_{i=1}^{m} \cut{\unrollP{\fullPart{A_i}}}{n} \cdot \forest{Y}^i = \cut{\unrollP{\fullPart{B}}}{n}$ using algorithm~ of Section~5 of~\cite{Antonio2026},
		\If{no solution exists}
		\State \Return $\perp$,
		\EndIf
		\State let $\forest{T}$ be the solution,
		\State \textbf{// Extract nilpotent part}
		\State $\forest{F} \gets \forest{X} - \forest{T}$,
		\State \Return \texttt{InverseUnroll}$(\forest{F})$.
	\end{algorithmic}
\end{algorithm}

Since all the operations described above can be performed in polynomial time, as the chosen $n$ is polynomial in the size of $B$, we conclude the following result.

\begin{prop}\label{prop:sol_nil}
	Let $P = \sum_{i=1}^{m} A_i X^i$ be a polynomial over partial systems without constant term such that there exists $i$ with $A_i$ containing a cycle.
	If there exists $X$ such that $P(X) = B$, then we can construct $\nilPart{X}$ in polynomial time.
\end{prop}


We can now conclude that, if we have a polynomial-time algorithm to solve equations over total systems, then we have a polynomial-time algorithm to solve equations over partial systems.

\begin{Theorem}
	Let $P = \sum_{i=1}^{m} A_i X^i$ be a polynomial over partial systems without constant term such that there exists $i$ with $A_i$ containing a cycle.
	If we can efficiently solve the equation $\sum_{i=1}^{m} \fullPart{A_i} \fullPart{X}^i = \fullPart{B}$, then we can solve $P(X) = B$ in polynomial time.
    
\end{Theorem}

\begin{proof}
	We show that there exists $X$ a solution to $P(X) = B$ if and only if $P(\forest{X} + Y) = B$ with $\forest{X}$ the forest described in the discussion preceding Proposition~\ref{prop:sol_nil} and $Y$ a solution to $\sum_{i=1}^{m} \fullPart{A_i} \fullPart{X}^i = \fullPart{B}$.

	$(\Leftarrow)$ Suppose that $P(\forest{X} + Y) = B$.
	Then, directly, there exists $X$ such that $P(X) = B$.

	$(\Rightarrow)$ Suppose that there exists $Z$ such that $P(X) = B$.
	Then, by Lemma~\ref{lemma:inj_des_finis}, it follows that all solutions of $P(X) = B$ are of the form $\nilPart{Z} + Y$ with $Y$ an FDDS.

	Furthermore, since there exists $i$ such that $\fullPart{A_i} \neq \emptyset$, we also have that $Y$ is a solution to $\sum_{i=1}^{m} \fullPart{A_i} \fullPart{X}^i = \fullPart{B}$.
	We now show that all solutions $Y$ of $\sum_{i=1}^{m} \fullPart{A_i} \fullPart{X}^i = \fullPart{B}$ are such that $P(\nilPart{Z} + Y) = B$.

    Due to the choice of $Y$, we have $\sum_{i=1}^{m} \fullPart{A_i} \fullPart{Z}^i = \sum_{i=1}^{m} \fullPart{A_i} Y^i$. 
    Therefore, $\sum_{i=1}^{m} \unroll{\fullPart{A_i}} \unroll{\fullPart{Z}}^i = \sum_{i=1}^{m} \unroll{\fullPart{A_i}} \unroll{Y}^i$. 
    Consequently, Theorem~\ref{th:injPolyUnrolls} implies that $\unroll{\fullPart{Z}} = \unroll{Y}$ and thus $\unrollP{\fullPart{Z}} = \unrollP{Y}$.

	Let $i>0$ be an integer.
	We show that $(\nilPart{Z} + Y)^i - Y^i = X^i - \fullPart{Z}^i$.
	Equivalently, we prove that
	\[
	\sum_{j=0}^{i-1} \binom{i}{j} \nilPart{Z}^j \times Y^{i-j} = \sum_{j=0}^{i-1} \binom{i}{j}\nilPart{Z}^j \times \fullPart{Z}^{i-j}.
	\]
	For this, observe that from Lemma~\ref{lemma:caracterisation_nilpotent_full_product}, it follows that $\nilPart{Z}^j \times Y^{i-j} = \nilPart{Z}^j \treeProd \unrollP{Y}^{i-j}$.
	Similarly, $\nilPart{Z}^j \times \fullPart{Z}^{i-j} = \nilPart{Z}^j \treeProd \unrollP{\fullPart{Z}}^{i-j}$.
    we have
	\[
	\nilPart{Z}^j \times Y^{i-j} = \nilPart{Z}^j \times \fullPart{Z}^{i-j}.
	\]
	Thus, the claim follows by summation.

	At this point, we deduce that $\nilPart{P(Z)} = \nilPart{P(\nilPart{Z} + Y)}$ and $\fullPart{P(Z)} = \fullPart{P(\nilPart{Z} + Y)}$.
	Therefore, $P(Z) = P(\nilPart{Z} + Y)$.
\end{proof}

	\subsection{Solving $P(X) = B$ when all coefficients are finite forests} \label{section:division_in_Nilpotent}

Let $P = \sum_{i=1}^{m} \forest{A}_i \times X^i$ be a polynomial where all the $\forest{A}_i$ coefficients are finite forests and let $B$ be a partial system.
In this subsection, we show that we can solve the equation $P(X) = B$ in polynomial time with respect to the sum of the sizes of the $\forest{A}_i$ and $B$.
Remark that if there exists a partial system $X$ such that $P(X) = B$, then $\fullPart{B} = \emptyset$.
This is why, in the rest of this section, we further assume that $B$ is a finite forest, which we will therefore denote $\forest{B}$.

Since all coefficients of $P$ and $\forest{B}$ are finite forests, we would like to use algorithm~ of Section~5 of~\cite{Antonio2026} which solves polynomial equations over finite forests in polynomial time.

However, since this algorithm is based on polynomials whose product operation is $\treeProd$ and not $\times$, we need to modify the polynomial $P$.
Two approaches are possible.
Either we modify $P$ via the expressions of Lemmas~\ref{lemma:caracterisation_nilpotent_product} and~\ref{lemma:caracterisation_nilpotent_full_product}, or we apply the partial-system unroll to $P$.
We choose here to employ the second method.
To argue the correctness of this approach, we present the following result.

\begin{Lemma}
	Let $P = \sum_{i=1}^{m} \forest{A}_i \times X^i$ be a polynomial over finite forests.
	Consider $Q = \sum_{i=1}^{m} \unrollP{\forest{A}_i} \treeProd \forest{X}^i$ the polynomial over partial-system unrolls associated with $P$.
	Let $Y$ be a partial system and $\forest{B}$ a finite forest.
	Then, $P(Y) = \forest{B}$ if and only if $Q(\unrollP{Y}) = \unrollP{\forest{B}}$.
\end{Lemma}

\begin{proof}
	$(\Rightarrow)$ Suppose that $P(Y) = \forest{B}$.
	Then, by Lemma~\ref{lemma:unroll_distributivity_over_product_partial}, it follows that $\sum_{i=1}^{m} \unrollP{\forest{A}_i} \treeProd \unrollP{Y}^i = \unrollP{\forest{B}}$.
	Therefore, $Q(\unrollP{Y}) = \unrollP{\forest{B}}$.

	$(\Leftarrow)$ Suppose that $Q(\unrollP{Y}) = \unrollP{\forest{B}}$.
	Then, we can partition $\unrollP{\forest{B}}$ into $\unrollP{\forest{B}}_1 + \cdots + \unrollP{\forest{B}}_m$ such that $\unrollP{\forest{A}_i} \treeProd \unrollP{Y}^i = \unrollP{\forest{B}}_i$ for all $i$ and $\sum_{i=1}^{m} \unrollP{\forest{B}}_i = \unrollP{\forest{B}}$.
    Since, by Lemma~\ref{lemma:unroll_distributivity_over_product_partial},we have  $\unrollP{\forest{A}_i} \treeProd \unrollP{Y}^i = \unrollP{\forest{A}_i \times Y^i}$ for all $i$, it follows that 
    \[\sum_{i=1}^{m} \unrollP{\forest{A}_i \times Y^i} = \mathfrak{U}\bigg(\sum_{i=1}^{m} \forest{A}_i \times Y^i\bigg)  = \unrollP{\forest{B}}.\] 

    Finally, since $\forest{A}_i$ is a finite forest for all $i$, it is also the case for $\forest{A}_i \times Y^i$ for all $i$.
	Thus Proposition~\ref{proposition:unroll_is_isomorphism}, implies that $\sum_{i=1}^{m} \forest{A}_i \times Y^i = \forest{B}$.
	We conclude that $P(X) = \forest{B}$.
\end{proof}

Thus, using algorithm~ of Section~5 of~\cite{Antonio2026}, we solve $Q(\forest{X}) = \forest{B}$ in polynomial time; let $\forest{F}$ denote the solution found.
Remark that $\forest{F}$ is a forest, cut at depth $\depth{P}$, the maximum depth of the $\depth{\forest{A}_i}$.
Furthermore, $\forest{F}$ is not necessarily the partial-system unroll of a partial system, as shown in Figure~\ref{fig:contre_example_divison_unroll}.

\begin{figure}[t]
	\centering
	\includegraphics[page=3,scale=0.84]{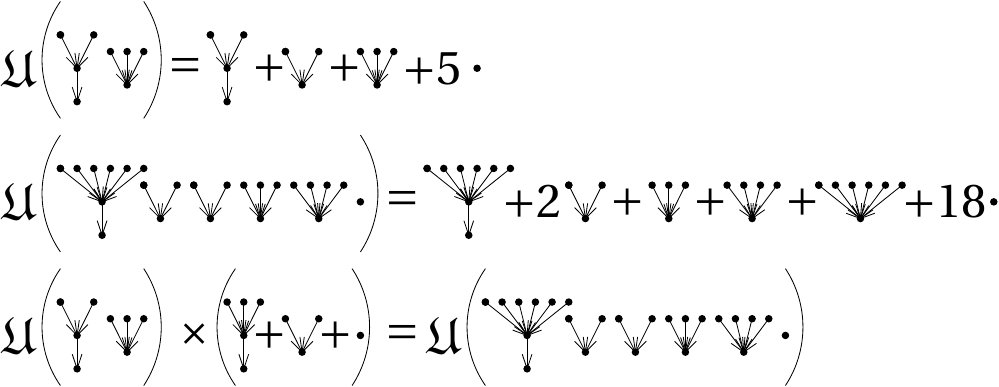}
	\caption{Example of division of a partial-system unroll with a quotient that is not the cut of a partial-system unroll.}
	\label{fig:contre_example_divison_unroll}
\end{figure}

Observe that two cases are possible.
First case: $\depth{B} < \depth{P}$.
Then, Proposition~\ref{prop:injDesFinis} implies that $\forest{F}$ is the unique forest such that $Q(\forest{F}) = \forest{B}$.
Therefore, we solve $P(X) = \forest{B}$ by reconstructing, if possible, the forest $\forest{X}$ such that $\unrollP{\forest{X}} = \forest{F}$.
This can be done in polynomial time with respect to the size of $\forest{F}$, and hence with respect to the size of $\forest{B}$, as explained in the discussion preceding Proposition~\ref{theorem:caracterization_of_counter-example_of_cancelation}.
The second case is $\depth{B} = \depth{P}$. 
The difficulty in this case is that even if there exists a partial-system $X$ such that $P(X) = \forest{B}$, we have $\unrollP{X}$ might not be equal to $\forest{F}$. 
Indeed, $\forest{F} = \cut{\unrollP{X}}{\depth{P}}$  and we need to explain how we reconstruct the partial system $X$.
This is the subject of what follows.

Let $X$ be a partial system and $d\ge 1$ an integer.
Consider $\tree{t}$ a tree of $\unrollP{X}$.
In order to infer properties that would allow us to construct a partial system from a forest $\forest{F}$ of depth $d$, we analyze the elements that constitute $\unrollP{X}$.
Remark that two cases are then possible.

Either $\tree{t}$ is rooted in a vertex of depth $0$ in $\unrollP{X}$, and thus $\tree{t}$ is a tree of $\nilPart{X} + \unroll{\fullPart{X}}$, or this is not the case.
In the latter case, the root of $\tree{t}$ has a successor $u$ in $X$.
Again, two cases are possible.
The vertex $u$ may be a vertex of a cycle of $X$, in which case $\tree{t}$ is a tree of $\dt{\unroll{\fullPart{X}}}$.
Or this is not the case, implying that there exists a finite tree $\tree{t}'$ rooted in $u$ containing all predecessors of $u$ in $X$.
And since, by definition of the partial-system unroll, the tree $\tree{t}'$ belongs to $\unrollP{X}$, it follows that $\tree{t}$ is in $\dt{\unrollP{X}}$.
From these last two cases, we deduce that there exists a tree $\tree{t}'$ of $\unrollP{X}$ such that $\dt{\cut{\tree{t}'}{d}}$ contains $\cut{\tree{t}}{d-1}$.
From this reasoning, we deduce the following lemma.

\begin{Lemma}\label{lemma:closure}
	Let $d\ge1$ be an integer and let $\forest{F}$ be a finite forest of depth at most $d$.
	If there exists a partial system $X$ such that $\cut{\unrollP{X}}{d} = \forest{F}$, then $\dt{\forest{F}}$ is a submultiset of $\cut{\unrollP{X}}{d-1}$.
\end{Lemma}

We now show that the converse implication also holds.
For this, we introduce Algorithm~\ref{algo:roll_partiel} which proceeds by ``gluing'' trees onto connected components.
In order to formally define this \emph{gluing} operation, we need to introduce some preliminary notions.

\begin{Definition}[Upper and lower parts]
	A \emph{super-leaf} of depth $k$ of a partial system $A$ is the trees of $\dt{A}$ having depth $A$. 

	The multiset of \emph{upper parts} of a partial system $A$ is, in our context, the multiset of super-leaves (of any depth $k$) rooted in the vertices of $A$ at depth $0$.
	The multiset of \emph{lower parts} of $A$ is $\cut{\nilPart{A}}{d-1}$ with $d$ the depth of $A$.
\end{Definition}
Remark that, in other words, the upper parts of a partial system are $\dt{A}$.
Figure~\ref{fig:partie_b_et_h} gives an example of super-leaves, lower parts and upper parts.

\begin{Definition}[Gluing and ungluing]\label{def:collage}
	Let $\tree{t}$ be a tree and $X$ a connected PDDS such that the lower part of $\tree{t}$ is equal to an upper part $h$ of $X$. 
	We call \emph{gluing} of $\tree{t}$ onto $X$ by $h$ the connected partial system obtained by replacing $h$ with $\tree{t}$ in $X$.
	When a tree is glued on an upper part, we say that the upper part has been \emph{filled}.

\end{Definition}

\begin{figure}[t]
	\centering
	\includegraphics[page=4]{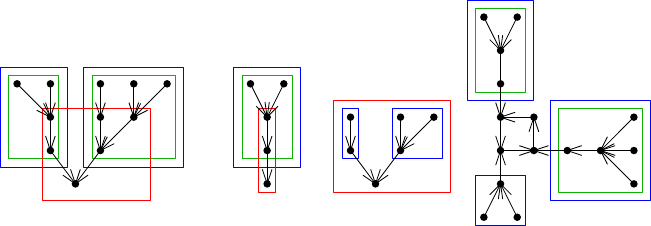}
	\caption{Illustrations of the notions of super-leaves, upper parts and lower parts of a partial system.
		The parts outlined in red are the lower parts, those in blue are the upper parts and those in green are the super-leaves of depth $2$.}\label{fig:partie_b_et_h}
\end{figure}

\begin{figure}[t]
	\centering
	\includegraphics[page=5]{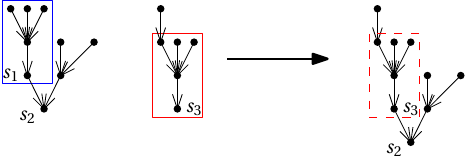}
	\caption{Illustration of a gluing operation.
		We glue the tree whose lower part is outlined in red onto the tree whose upper part is outlined in blue.
		The result is the tree to the right of the arrow.}\label{fig:collage1}
\end{figure}
Figure~\ref{fig:collage1} gives an example of a gluing operation.

Like previously said, the goal of Algorithm~\ref{algo:roll_partiel} is to start from a forest $\forest{F}$ of depth $d$ and to build, by a succession of gluings, a PDDS $X$ such that $\cut{\unrollP{X}}{d} = \forest{F}$.
It starts by computing two multisets, namely $\mathcal{H}$ and $\mathcal{B}$, respectively the multiset of upper parts of $\forest{F}$ and the multiset of its lower parts.
After that, it checks that $\mathcal{H}$ is a submultiset of $\mathcal{B}$, according to Lemma~\ref{lemma:closure}.
Then, it defines a variable $X$ initially containing $\forest{F}$.
Now, we enter the core of the procedure.
It takes an element $h$ of $\mathcal{H}$ and $b$ of $\mathcal{B}$ such that $h = b$, the parts it glues.
Then, it selects a tree $\tree{t}$ of $X$ having lower part $b$ and a connected component $C$ of $X$ having upper part $h$.
It fills $h$ with $b$, forming a connected component $C'$, and removes $h$ from $\mathcal{H}$, $b$ from $\mathcal{B}$, and $\tree{t}$ and $C$ from $X$.
It repeats these steps while $\mathcal{H}$ is not empty.
Figure~\ref{fig:run_roll_partiel} shows an example of a run of this algorithm.

\begin{figure}[p]
	\centering
	\begin{adjustbox}{trim=0 {6.1cm} 0 0, clip, width=\textwidth}
		\includegraphics[page=8]{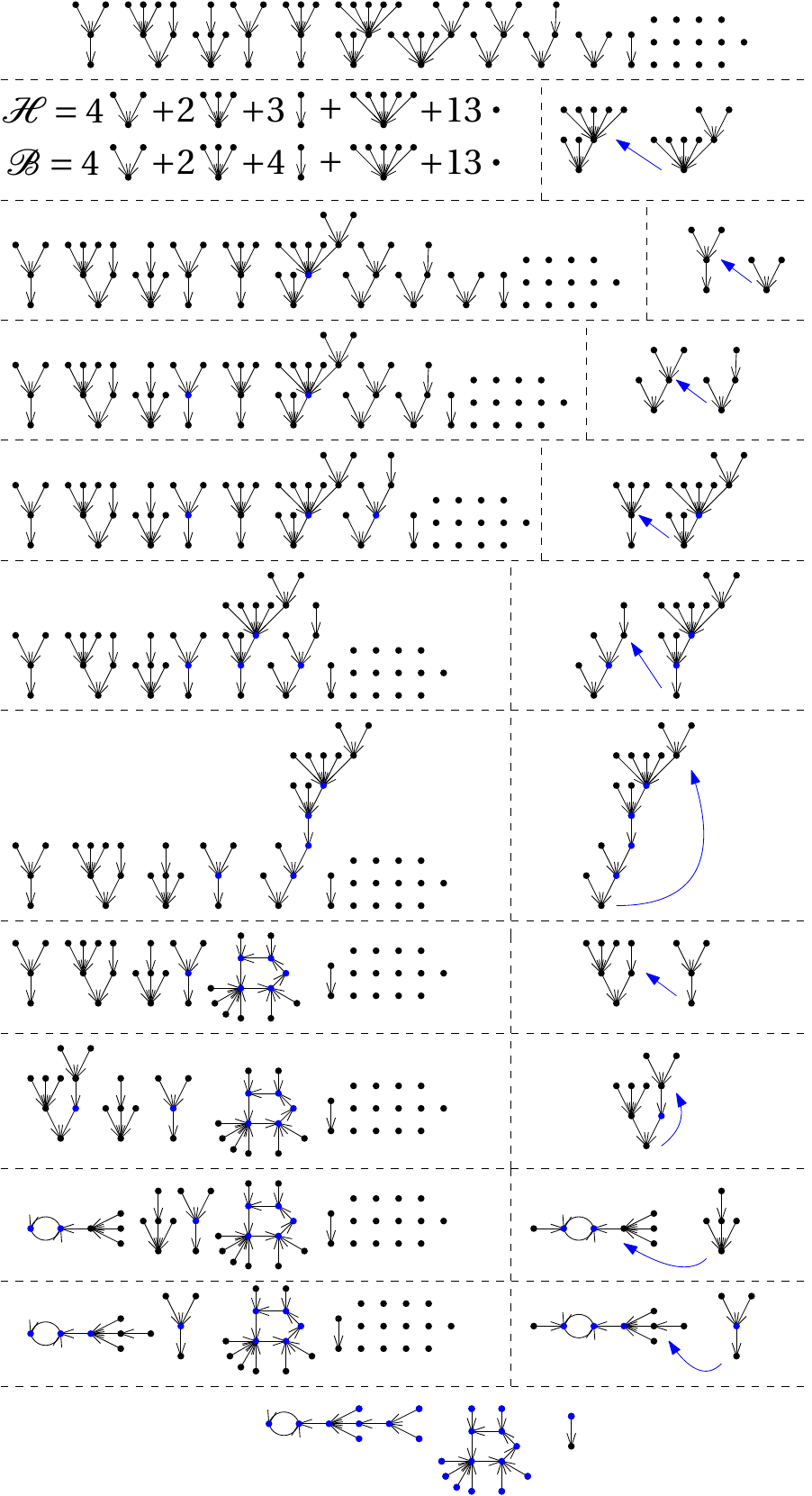}
	\end{adjustbox}
\end{figure}

\begin{figure}[t]
	\centering
	\begin{adjustbox}{trim=0 0 0 {22cm}, clip, width=\textwidth}
		\includegraphics[page=8]{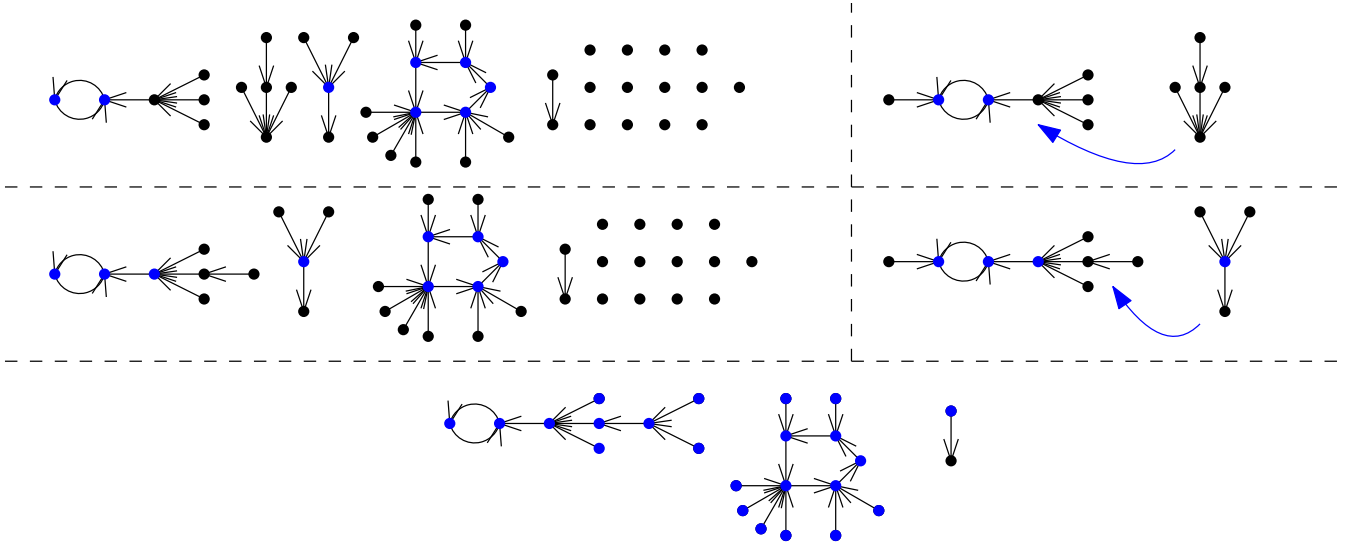}
	\end{adjustbox}
	\caption{Example of execution of Algorithm~\ref{algo:roll_partiel} on the forest given at the top of the page.
		The result is the partial system at the bottom of the page.
		Each iteration of the algorithm is separated by a dashed line.
		The blue vertices are those that were glued during the different iterations.
		To go from the second-to-last line to the final result, we simply glued all the leaves.}
	\label{fig:run_roll_partiel}
\end{figure}

\begin{algorithm}[t]
	\caption{\texttt{PartialRoll}$(\forest{F})$}\label{algo:roll_partiel}
	\begin{algorithmic}[1]
		\State let $\mathcal{H}=\dt{\forest{F}}$ and $\mathcal{B}= \cut{\forest{F}}{d-1}$,
		\If{$\mathcal{H}$ is not a submultiset of $\mathcal{B}$}
		\State \Return $\perp$,
		\EndIf
		\State $X \gets \forest{F}$,
		\While{$\mathcal{H}$ is not empty}\label{algo:roll_partiel:loop}
		\State let $h$ be an element of $\mathcal{H}$ and $C$ its associated connected component,
		\State\label{algo:roll_partiel:C} remove $h$ from $\mathcal{H}$ and $C$ from $X$,
		\State\label{algo:roll_partiel:t} let $b$ be an element of $\mathcal{B}$ equal to $h$ and $\tree{t}$ the finite tree associated with it,
		\State remove $b$ from $\mathcal{B}$ and $\tree{t}$ from $X$,
		\State\label{algo:roll_partiel:C'} let $C'$ be the connected component resulting from the gluing of $\tree{t}$ onto $C$ by $h$,
		\State add $C'$ to $X$,
		\EndWhile\label{algo:roll_partiel:Endloop}
		\State \Return $X$.
	\end{algorithmic}
\end{algorithm}

We show that Algorithm~\ref{algo:roll_partiel} is correct.
In other words, we prove that it returns a PDDS if and only if there exists a PDDS $X$ such that $\cut{\unrollP{X}}{d} = \forest{F}$, with $\forest{F}$ the input forest and $d$ the depth of $\forest{F}$.
For this, observe that it is sufficient to prove that:
\begin{enumerate}
    \item the algorithm terminates, and
    \item the PDDS returned respects the expected condition ($\cut{\unrollP{X}}{d} = \forest{F}$).
\end{enumerate}
Let us first discuss termination.

Observe that an easy induction proof allows us to show that, at the $i$-th iteration of the loop spanning lines~\ref{algo:roll_partiel:loop}--\ref{algo:roll_partiel:Endloop} of Algorithm~\ref{algo:roll_partiel}, the multiset $\mathcal{H}$ contains all the upper parts of the original $\forest{F}$ that are not yet filled.

Note, in addition, that at the $i$-th iteration of this loop, the multiset $\mathcal{B}$ contains $\cut{\nilPart{X}}{d-1}$, with $X$ the current result, since we have already removed the lower part that we just glued, and since only trees have a lower part.

Thanks to these observations, we deduce that an induction proof allows us to prove that, after each iteration of the loop, $\mathcal{H}$ is a submultiset of $\mathcal{B}$, and additionally that $\mathcal{B} = \cut{\nilPart{X}}{d-1}$, with $X$ the PDDS under construction. 
In other words, at the end of each iteration, we can glue a tree of $\nilPart{X}$ onto each element of $\mathcal{H}$.

We conclude that we can perform each operation of the loop of Algorithm~\ref{algo:roll_partiel}, and therefore that it terminates.
In addition, given that each operation in the loop can be done in polynomial time, and that the number of iterations is polynomial in the number of trees in $\forest{F}$, it follows that the algorithm runs in polynomial time.

Thus, as previously stated, to prove Theorem~\ref{th:roll_partiel_poly}, it is sufficient to establish that the condition $\cut{\unrollP{X}}{d} = \forest{F}$ is satisfied.
This is the purpose of Theorem~\ref{th:unroll_partiel_poly}.

\begin{Theorem}\label{th:roll_partiel_poly}
	Algorithm~\ref{algo:roll_partiel} runs in polynomial time with respect to the size of its inputs.
	Thus, solving $P(\forest{X}) = \forest{B}$ is in polynomial time with respect to the sum of the sizes of the $\forest{A}_i$ and $\forest{B}$ where $P = \sum_{i=1}^{m} \forest{A}_i \times \forest{X}^i$ is a polynomial over finite forests and $\forest{B}$ is a finite forest.
\end{Theorem}

\begin{Theorem}\label{th:unroll_partiel_poly}
	Let $d\ge1$ be an integer and let $\forest{F}$ be a finite forest of depth at most $d$.
	Then, there exists a partial system $X$ such that $\cut{\unrollP{X}}{d} = \forest{F}$ if and only if $\dt{\forest{F}}$ is a submultiset of $\cut{\forest{F}}{d-1}$.
\end{Theorem}

\begin{proof}
	Let $X$ be the partial system resulting from the execution of Algorithm~\ref{algo:roll_partiel} on $\forest{F}$.
	The construction of $X$ implies that its vertex set corresponds to the set of roots of the trees of $\forest{F}$.
	Indeed, a vertex of $X$ is either a root and hence also a root of $\forest{F}$, or the root of a tree associated with a lower part of $\forest{F}$ that was glued onto an upper part of $\forest{F}$.
	For brevity, such roots are said to be glued in $X$.

    Let $i\ge 0$ be an integer and let $\forest{R}_i(X)$ be the trees of $\unrollP{X}$ whose root was glued in $X$ during an iteration $j$ such that $j \le i$.
	Finally, let $\forest{S}_i(X) = \nilPart{X} + \forest{R}_i(X)$. 
    Note that, when $i$ is greater than the number of iterations of the loop in lines~\ref{algo:roll_partiel:loop}--\ref{algo:roll_partiel:Endloop} of Algorithm~\ref{algo:roll_partiel}, then $\forest{S}_i(X) = \unrollP{X}$. 
    
	We prove by induction on the number of iterations that $\cut{\forest{S}_i(X)}{d} = \forest{F}$.

	If no gluing has been performed (case $i=0$), then $\forest{S}_i(X) = \forest{F}$ since the value of $\nilPart{X}$ at the $0$-th iteration is $\forest{F}$.
	Let $i\ge0$ be an integer.
	Suppose that the property holds for $i$.
	We show that it holds for $i + 1$.
	Let $C$ be the connected component and $\tree{t}$ the tree as defined on lines~\ref{algo:roll_partiel:C} and~\ref{algo:roll_partiel:t} of Algorithm~\ref{algo:roll_partiel} respectively.
	Let $t$ be the root of $\tree{t}$ and let $s$ be the vertex of $C$ in which we root $\tree{t}$, \ie, the vertex $s_2$ from Definition~\ref{def:collage}.
	Remark that the set of vertices of $\cut{\forest{S}_{i+1}(X)}{d}$ not containing any vertex of the form $(s,j)$ with $0 \le j \le d$ remains unchanged between $\forest{S}_{i}(X)$ and $\forest{S}_{i+1}(X)$.
	Indeed, the modification does not affect them since either the vertex $s$ is not in their connected component, or the distance between $s$ and them is strictly greater than $d$.

	Furthermore, by definition of the gluing operation, the lower part of $\tree{t}$ is equal to the upper part of $C$ that we removed.
	Therefore, the unroll tree $\tree{s}$ rooted in $(s,0)$ and cut at depth $d$ is isomorphic between the two iterations, implying that the same holds for the trees $\cut{\tree{s}}{j}$ for every integer $j$ between $0$ and $d$.
	The property is thus hereditary and $\cut{\unrollP{X}}{d} = \forest{F}$.
\end{proof}

\begin{Theorem}
	Let $d \ge 1$ be an integer and let $\forest{F}$ be a finite forest of depth at most $d$. 
	Then, there exists an FDDS $X$ such that $\cut{\unrollP{X}}{d} = \forest{F}$ if and only if $\dt{\forest{F}} = \cut{\forest{F}}{d-1}$.
\end{Theorem}

\begin{proof}
	$(\Rightarrow)$ Let us assume that there exists an FDDS $X$ such that $\cut{\unrollP{X}}{d} = \forest{F}$. 
	Then, by Theorem~\ref{th:unroll_partiel_poly} we have that $\dt{\forest{F}}$ is a submultiset of $\cut{\forest{F}}{d-1}$. 
	To prove that $\dt{\forest{F}} = \cut{\forest{F}}{d-1}$, it remains to show the reverse inclusion. 
	For this, assume by contradiction that $\dt{\forest{F}}$ is a proper submultiset of $\cut{\forest{F}}{d-1}$.
	Then, there exists an element of $\cut{\forest{F}}{d-1}$ whose multiplicity is greater than in $\dt{\forest{F}}$.
	Therefore, we cannot glue each element of $\cut{\forest{F}}{d-1}$ to an element of $\dt{\forest{F}}$.
	We conclude that $\nilPart{Y} \neq \emptyset$ for each possible partial system $Y$ such that $\cut{\unrollP{Y}}{d} = \forest{F}$.
	This implies that $X$ is not an FDDS, a contradiction.

	$(\Leftarrow)$ Let us assume that $\dt{\forest{F}} = \cut{\forest{F}}{d-1}$.
	Therefore, by Theorem~\ref{th:unroll_partiel_poly}, there exists a partial system $X$ such that $\cut{\unrollP{X}}{d} = \forest{F}$.
	We prove that $X$ is an FDDS.
	For this, note that if $\nilPart{X} \neq \emptyset$, then there exists an element of $\cut{\forest{F}}{d-1}$ that is not glued to any element of $\dt{\forest{F}}$.
	But since each element of $\dt{\forest{F}}$ is glued to an element of $\cut{\forest{F}}{d-1}$ in $X$, this implies that some element of $\cut{\forest{F}}{d-1}$ has greater multiplicity than in $\dt{\forest{F}}$, a contradiction.
\end{proof}

\subsection{Solving $P( \forest{X} ) = \forest{B}$}

While Algorithm~\ref{algo:roll_partiel} solves, in polynomial time with respect to the size of the inputs, equations of the form $\sum_{i=0}^{m} \forest{A}_i X^i = \forest{B}$ in $(\Partial,+,\times)$, it does not do so if we restrict ourselves to the pseudo-semiring $(\mathbb{F}_f,+,\times)$.
Indeed, some finite forests can be the partial-system unroll of a partial system with cycles without being that of a forest.

\begin{Example}
	Consider the forest consisting solely of a path of depth at least $1$.
	Then this forest can be rolled into a fixed point.
	However, it is not the unroll of a tree since the unroll of a tree of depth at least $1$ contains at least two connected components.
\end{Example}

We will therefore now refine our approach to search only for solutions that are finite forests.

In order to cover this case, we construct a graph that models the different possible gluings.
In other words, given a forest $\forest{F}$, we construct a graph where each vertex represents a tree $\tree{t}$ of $\forest{F}$ and is composed of a lower port, labeled by the lower part of $\tree{t}$, and several upper ports, labeled by the upper parts of $\tree{t}$.
See Figure~\ref{fig:graphe_de_collage} for an example.
We add edges between a lower port and an upper port if and only if they have the same label.
This graph is called the \emph{gluing graph} of $\forest{F}$.
An example of such a graph is given in Figure~\ref{fig:graphe_de_collage}.
Remark that in this graph, some vertices have no upper port.
In other words, they are the ``leaves'' of this graph.
Our goal is now to prove the following theorem.

\begin{Theorem}\label{th:cara_sol_arbre}
	Let $\forest{F}$ be a finite forest of depth $d$.
	And let $\mathcal{H}$ and $\mathcal{B}$ be the multisets of upper and lower parts respectively.
	Then, there exists a finite forest $\forest{X}$ such that $\cut{\unrollP{\forest{X}}}{d} = \forest{F}$ if and only if the following three conditions are satisfied:
	\begin{enumerate}
		\item\label{p1} $\mathcal{H}$ is a submultiset of $\mathcal{B}$,
		\item\label{p2} there exist $b_1,\ldots,b_n$ elements of $\mathcal{B}$ having strictly greater multiplicity in $\mathcal{B}$ than in $\mathcal{H}$,
		\item\label{p3} for every vertex $u$ of the gluing graph of $\forest{F}$, there exists a path to a vertex $v$ whose lower port is one of the $b_i$.
	\end{enumerate}
\end{Theorem}

In order to prove this theorem, we introduce the notion of \emph{gluing plan} of a graph.
A gluing plan of a gluing graph is a subgraph containing all vertices of the gluing graph and in which each upper port has exactly one incoming edge and each lower port has at most one outgoing edge.
We will show that each possible solution of Algorithm~\ref{algo:roll_partiel} describes a gluing plan and vice versa.
An example is given in Figure~\ref{fig:graphe_de_collage}.

\begin{figure}[t]
	\centering
	\includegraphics[page=6]{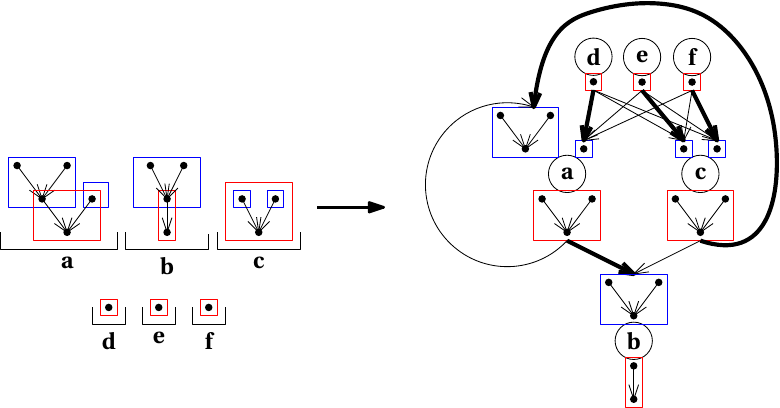}
	\caption{Example of construction of a gluing graph.
		We start from the trees to the left of the central arrow.
		Each tree outlined in blue is an upper part and those outlined in red are the lower parts.
		Each vertex of the graph on the right models the tree whose name is at the center of the circle.
		The bold edges describe a gluing plan.}\label{fig:graphe_de_collage}
\end{figure}

\begin{Lemma}\label{lemma:sol_si_plan}
	Let $\forest{F}$ be a forest of depth $d$.
	Then, there exists $X$ a PDDS such that $\cut{\unrollP{X}}{d} = \forest{F}$ if and only if the gluing graph of $\forest{F}$ admits a gluing plan.
	Furthermore, $X$ is a forest if and only if this gluing plan is acyclic.
	Finally, we can convert between the two representations in polynomial time.
\end{Lemma}

\begin{proof}
	$(\Rightarrow)$ Let $X$ be a partial system such that $\cut{\unrollP{X}}{d} = \forest{F}$.
	We explain how to construct a gluing plan from $X$.
	Consider $\tree{t}$ a tree of $\unrollP{X}$ with root $(u,0)$.
	For each such tree, we create a vertex with a label of lower port $\cut{\tree{t}}{d-1}$; we say this label has root $(u,0)$.
	Furthermore, this vertex will have one upper port per tree $\tree{t}'$, with root $(v,1)$, of $\dt{\tree{t}}$, with label $\cut{\tree{t}'}{d-1}$, and this label will have root $(v,0)$.
	From this, we deduce that the gluing graph of $\forest{F}$ admits a gluing plan.

	Consider the graph with vertex set constructed as described above from the trees of $\unrollP{X}$.
	Furthermore, we add an edge between a lower port and an upper port if and only if their labels have the same root.
	Therefore, the definition of partial-system unrolls implies that each upper port has in-degree exactly $1$ and each lower port has out-degree at most $1$.
	Furthermore, if $X$ is a forest, then this graph is also acyclic.
	Moreover, the trees in the labels of the two endpoints of each edge are isomorphic.
	Therefore, we take each port label up to isomorphism; the graph we have created is a gluing plan.

	Note that the operations described above can be performed in polynomial time and that the time to create the gluing graph of $\forest{F}$ is polynomial in the size of $\forest{F}$.
	Indeed, the cost of creating each vertex of the gluing graph of $\forest{F}$ is polynomial.
	And the number of ports (upper and lower combined) is bounded by the number of vertices of $\forest{F}$.
	Therefore, since checking that two labels are equal can be done in polynomial time, adding all edges of the gluing graph is also polynomial.

	$(\Leftarrow)$ Suppose now that the gluing graph of $\forest{F}$ admits a gluing plan.
	We can apply the inverse operation of the one described above to form a tree from a vertex of a gluing plan.
	Thus, if we perform the gluings described in the plan, \ie, we glue a lower part of a tree onto an upper part if and only if their associated ports are connected by an edge, a reasoning similar to that in the proof of Theorem~\ref{th:unroll_partiel_poly} implies that the partial system $X$ that we create satisfies $\cut{\unrollP{X}}{d} = \forest{F}$.
	Furthermore, if this plan is acyclic, then $X$ is also acyclic, since we perform no gluing that creates a cycle.
	Finally, the gluings can be performed in polynomial time.
\end{proof}

Building on Lemma~\ref{lemma:sol_si_plan}, we deduce that in order to find in polynomial time a forest $\forest{X}$ such that $\cut{\unrollP{\forest{X}}}{d} = \forest{F}$ where $\forest{F}$ is a forest of depth $d$, it suffices to find a gluing plan of the gluing graph of $\forest{F}$ and perform the gluings it describes.
In other words, we have a functional polynomial-time reduction between the two problems.
Furthermore, this same lemma implies that in order to show that if there exists $\forest{X}$ such that $\cut{\unrollP{\forest{X}}}{d} = \forest{F}$ where $\forest{F}$ is a forest of depth $d$, then the three conditions of Theorem~\ref{th:cara_sol_arbre} are satisfied.

\begin{Lemma}\label{th:cara_sol_arbre:cond_nec}
	Let $\forest{F}$ be a finite forest of depth $d$.
	And let $\mathcal{H}$ and $\mathcal{B}$ be the multisets of upper and lower parts respectively.
	If there exists $\forest{X}$ such that $\cut{\unrollP{\forest{X}}}{d} = \forest{F}$, then the three conditions of Theorem~\ref{th:cara_sol_arbre} are satisfied.
\end{Lemma}

\begin{proof}
	Suppose that there exists $\forest{X}$ such that $\cut{\unrollP{\forest{X}}}{d} = \forest{F}$.
	Then, Lemma~\ref{lemma:sol_si_plan} implies that the gluing graph of $\forest{F}$ admits a gluing plan.
	We show the three conditions in order.
	\begin{enumerate}
		\item By definition of a gluing plan, each upper port has as predecessor a unique lower port with the same label, and this predecessor has only one successor.
		This implies that the labels of upper ports (\ie, the upper parts of $\forest{F}$) are a submultiset of those of lower ports (\ie, the lower parts of $\forest{F}$).
		\item Also, since a gluing plan is acyclic and a lower port has only one upper port as successor, it follows that some lower ports (\ie, lower parts) have no successor.
		Let $b$ be such a label.
		Then, the multiplicity of $b$ in the multiset of lower parts of $\forest{F}$ is strictly greater than that in the multiset of upper parts of $\forest{F}$.
		\item From the previous point, the labels of the lower ports of the roots of the trees of the gluing plan have multiplicity in the multiset of lower parts of $\forest{F}$ that is strictly greater than that in the multiset of upper parts of $\forest{F}$.
		From there, since every vertex of the gluing plan admits a path to the roots of this plan, and since the gluing plan is a subgraph of the gluing graph with the same vertex set, it follows that every vertex $u$ of the gluing graph admits a path to one of the roots of the gluing plan. \qedhere
	\end{enumerate}
\end{proof}

Let $\forest{F}$ be a finite forest of depth $d$.
	Let $\mathcal{H}$ and $\mathcal{B}$ be the multisets of the upper and lower parts of $\forest{F}$, respectively.
	To show the converse direction, that is, that if the three conditions of Theorem~\ref{th:cara_sol_arbre} hold, then there exists a forest $\forest{X}$ such that $\cut{\unrollP{X}}{d} = \forest{F}$, we give an algorithm which builds an acyclic gluing plan if the three conditions hold.
	
	Our procedure starts from $G$ the gluing graph of $\forest{F}$ and begins by testing whether the three conditions of Theorem~\ref{th:cara_sol_arbre} hold.
	If this is not the case, it answers that no solution exists.
	This is correct by Lemma~\ref{th:cara_sol_arbre:cond_nec}.
	
	If these three conditions do hold, then the algorithm builds $E$ the set of lower port labels induced by $\mathcal{B} - \mathcal{H}$ and it applies Algorithm~\ref{algo:roll_partiel} to $\forest{F}$.
	Since condition~\ref{p1} holds, this call returns a partial system $X$ such that $\cut{\unrollP{X}}{d} = \forest{F}$.
	However, the partial system $X$ may contain cycles, so we have to remove them.
	
	For this, we introduce a subprocess that will allow us to swap arcs of $P$, the gluing plan induced by $X$, in order to remove its cycles.
	It takes as input $P$ as well as $s$ a vertex of a cycle of $P$ and $\rho$ a simple path of $G$ joining $s$ to a vertex having a lower port label included in $E$.
	It looks for the vertex $u$ of $\rho \cap C$ being the furthest from $s$, with $C$ the cycle of $P$ containing $s$, as well as the vertices $u'$, $v$ and $v'$ such that $v$ and $v'$ are the successors of $u$ in $\rho$ and $C$ respectively, while $u'$ is a predecessor of $v$ by an upper port with label $p$, with $p$ the label of the lower port of $u$.
	It replaces in $P$ the arcs from $u$ to $v'$ and from $u'$ to $v$ by arcs from $u$ to $v$ and from $u'$ to $v'$, respectively.
	Note that $v$ may be undefined if $u$ is the last vertex of $\rho$, in which case $u'$ is a vertex having no successor in $P$ and having a lower port label included in $E$.
	From there, we simply add an arc from $u'$ to $v'$ and delete the outgoing arc of $u$.
	
	We apply this procedure:
	\begin{enumerate}
		\item\label{c1} either from a new vertex $s'$ of a cycle of $P$ and a new path $\rho'$ if the number of cycles of $P$ has decreased through this swap,
		\item\label{c2} or from the vertex $u$ and the path $\rho$ if we are not in the previous case.
	\end{enumerate}
	An example of an execution of this algorithm is presented in Figure~\ref{fig:roll_partiel_tree_run}.

\begin{figure}[p]
    \centering
    \begin{adjustbox}{trim=0 {11cm} 0 0, clip, width=\textwidth}
        \includegraphics[page=9]{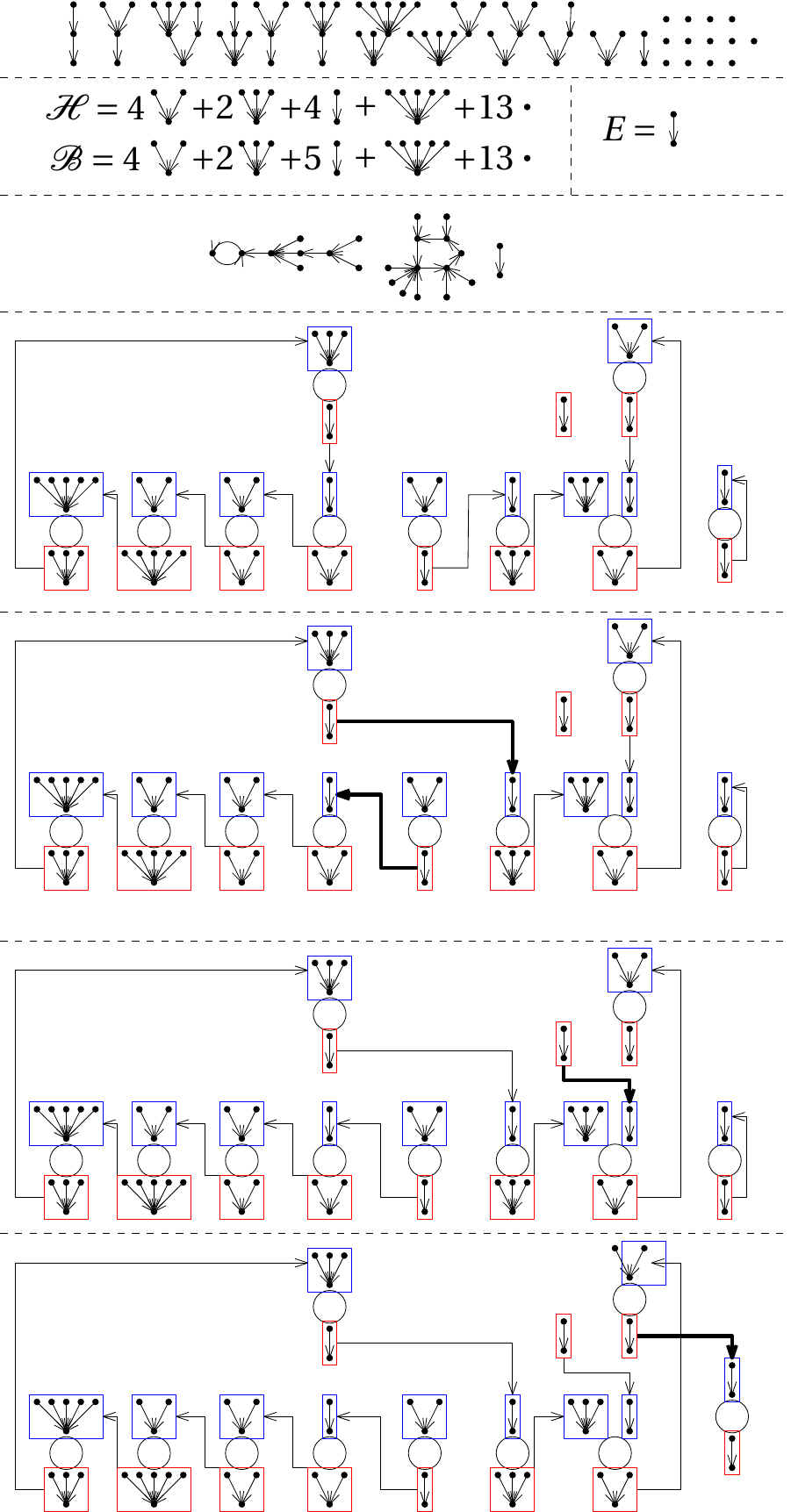}
    \end{adjustbox}
\end{figure}

\begin{figure}[t]
    \centering
    \begin{adjustbox}{trim=0 0 0 {18.2cm}, clip, width=\textwidth}
        \includegraphics[page=9]{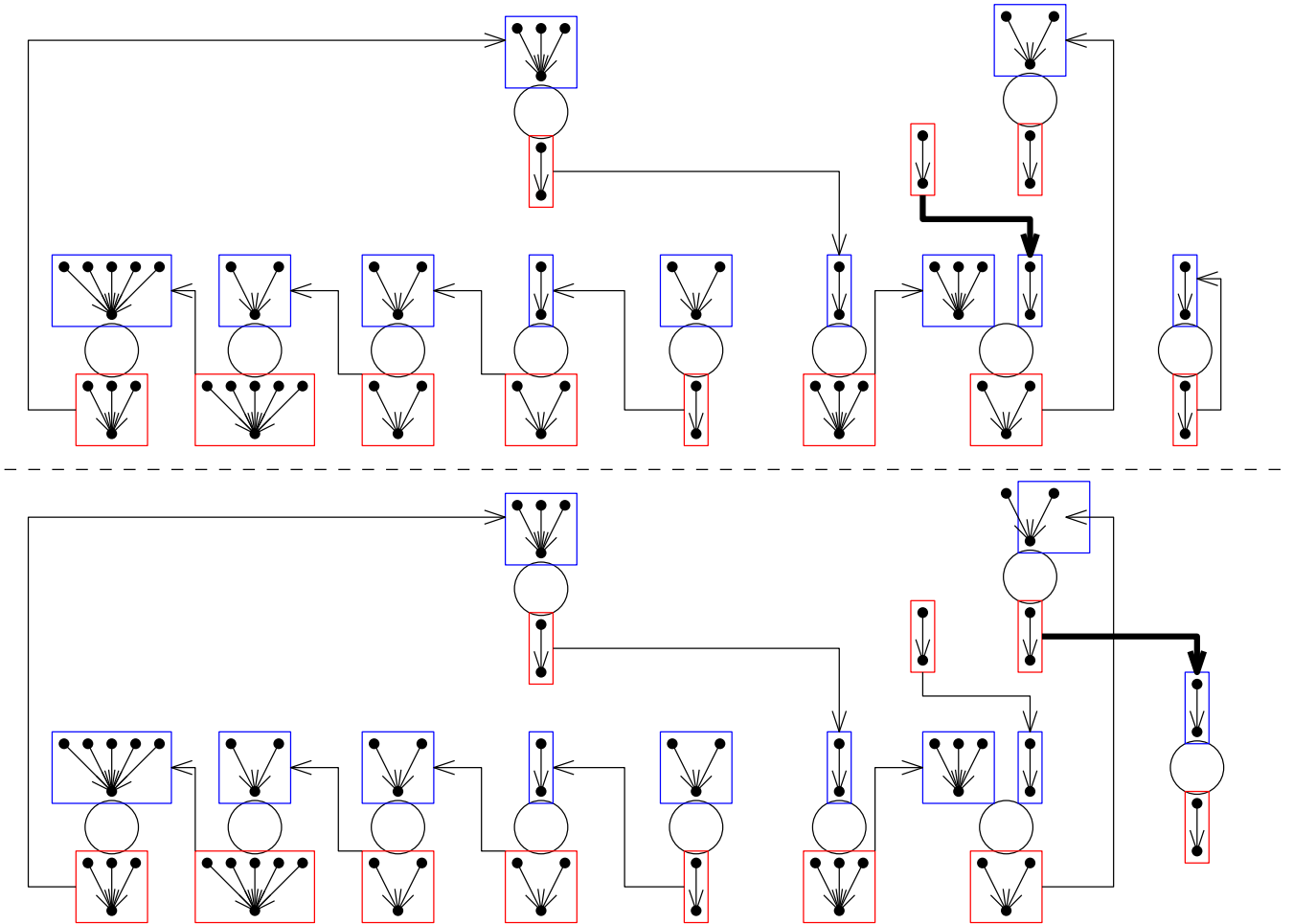}
    \end{adjustbox}
    \caption{Example of a run of the algorithm allowing rolling a finite forest in a new finite forest.
    The result is the gluing plan at the bottom of the page.
    Each step is separated by a dotted line. 
    The arcs in bold represent the modification carried out to incorporate the cycles in the trees.}
    \label{fig:roll_partiel_tree_run}
\end{figure}
	
\begin{Lemma}\label{lemma:algo_collage_correct}
		Let $G$ be the gluing graph of a forest $\forest{F}$.
		Our algorithm returns an acyclic gluing plan of $G$ in polynomial time if and only if one exists.
	\end{Lemma}
	
	\begin{proof}
		$(\Leftarrow)$ Suppose that the gluing graph admits an acyclic gluing plan.
		Then, by Lemma~\ref{lemma:sol_si_plan}, there exists a forest $\forest{X}$ such that $\cut{\unrollP{\forest{X}}}{d} = \forest{F}$, with $\forest{F}$ the forest induced by the gluing graph and $d$ the depth of $\forest{F}$.
		Thus, Lemma~\ref{th:cara_sol_arbre:cond_nec} implies that the three conditions of Theorem~\ref{th:cara_sol_arbre} are satisfied.
		Let us now show that the algorithm terminates.
		
		Note that at each step of the arc swapping procedure, we do not modify the in-degrees of the vertices $v$ and $v'$ and we do not increase the out-degrees of the vertices $u$ and $u'$.
		Thus, each iteration of the swapping procedure produces a gluing plan and we can indeed apply it recursively as in the algorithm.
		From there, to prove termination, it suffices to prove that there is a finite number of calls to this procedure.
		
		Let $P$ be a gluing plan of $G$, let $u,u',v,v'$ as well as $\rho$ be as in the swapping procedure and let $P'$ be the gluing plan produced at the end of the swapping procedure.
		Since the algorithm terminates when we have removed all the cycles of $P$, let us first show that the number of cycles does not increase.
		To do so, we prove that the result of the swapping procedure is a connected component, since a connected component contains at most one cycle.

		Note that if $v$ is well defined then $u$ and $v$ are not in the same cycle by the choice of $u$.
		This means that every simple path towards $u$ in $P$ is also present in $P'$, whether $v$ is defined or not, since we have modified arcs that belonged to none of these paths.
		Thus, since $v'$ is the successor of $u$ in a cycle of $P$, it follows that $P'$ contains a simple path between $v'$ and $u$ and also between $u'$ and $u$, since $P'$ contains an arc from $u'$ to $v'$.
		From there, if $v$ is not defined then the set of vertices having a simple path towards $u'$ in $P$ has a simple path towards $u$ in $P'$.
		This implies that the result is indeed connected.
		If now $v$ is well defined, all the simple paths towards $v$ in $P$ not going through the arc from $u'$ to $v$ are also present in $P'$.
		And since in $P'$ there exists a simple path from $u'$ to $u$ and therefore also from $u'$ to $v$, all the vertices connected to $v$ in $P$ by a simple path going through the arc from $u'$ to $v$ are also connected by a simple path towards $v$ in $P'$.
		In this case too the result is connected.
				
		We now prove that the successive application of the swapping procedure on the vertices of $\rho$ leads to a decrease in the number of cycles.
		For this, let us point out that three cases are possible depending on the existence of the vertex $v$ and on its position relative to $u$.
		
		The first case is that the vertex $v$ is well defined and that it is in the same connected component as $u$ in $P$.
		This means that there exists a simple path from $v$ to $u$ in $P$ and also in $P'$, as justified above.
		And since $P'$ contains the arc from $u$ to $v$, we conclude that $P'$ contains a cycle going through $u$ and $v$.
		This corresponds to case~\ref{c2} of the algorithm and we can apply this reasoning inductively either until $v$ is not in the same connected component as $u$, or until $u$ is the last vertex of $\rho$ and therefore $v$ is no longer defined.

		The second case is that $v$ is defined and that it is not in the same connected component as $u$ in $P$.
		Let us observe that two cases are again possible.
		Either the connected component containing $v$ contains a cycle.
		Since, as shown previously, the result of the swapping procedure is connected and since a connected component contains at most one cycle, it follows that the number of cycles has decreased by one.
		Or the connected component containing $v$ does not contain a cycle.
		Let us observe that it is not possible to create a cycle by deleting arcs.
		Hence, the removal of the arc from $u'$ to $v$ did not create a cycle, while the removal of the one from $u$ to $v'$ removed one.
		We therefore have to show that the addition of the arcs from $u'$ to $v'$ and from $u$ to $v$ does not create a cycle.
		Let us point out that since the connected component containing $v$ does not contain a cycle, the removal of the arc from $u'$ to $v$ produced two disjoint connected components.
		From there, since adding arcs between disjoint connected components cannot create a cycle, it follows that the addition of the two arcs does not create one.	 
		
		The third and last case is that $v$ is not defined.
		Note that the multiset of the lower port labels of the vertices of $P$ having no successor corresponds exactly to $\mathcal{B} - \mathcal{H}$.
		Hence, the vertex $u'$ defined in the arc swapping procedure does exist.
		From there, using the same kind of reasoning as for the previous case, it follows that the number of cycles decreases by one.
		
		We conclude that, starting from a cyclic connected component $D$ of $P$, the algorithm will produce a gluing plan $P'$ containing one cycle fewer than $P$ after at most $|\rho| \le |G|$ iterations.
		Hence, since after a decrease in the number of cycles we apply case~\ref{c1} of the algorithm, which starts again from a new cyclic connected component, we can apply this reasoning by induction and therefore the algorithm terminates after at most $|G|^2$ iterations and produces an acyclic gluing plan.
		Indeed, we have to consider at most $|G|$ simple paths of size at most $|G|$.
		
		$(\Rightarrow)$ 
		Suppose that our algorithm returns a graph.
		This then implies that it has terminated.
		Hence, by the previous point, it has produced an acyclic gluing plan.
	\end{proof}
	
	Lemmas~\ref{lemma:algo_collage_correct},~\ref{lemma:sol_si_plan} and~\ref{th:cara_sol_arbre:cond_nec} allow us to conclude the proof of Theorem~\ref{th:cara_sol_arbre}.
	Indeed, if the three conditions of Theorem~\ref{th:cara_sol_arbre} are satisfied, then our algorithm returns a graph which, by Lemma~\ref{lemma:algo_collage_correct}, is an acyclic gluing plan.
	For the other direction, if there exists an acyclic gluing plan, Lemmas~\ref{lemma:sol_si_plan} and~\ref{th:cara_sol_arbre:cond_nec} imply that the three conditions of Theorem~\ref{th:cara_sol_arbre} are satisfied.
	We thus conclude with the last result of this thesis.
	
	\begin{Theorem}
		Let $\forest{F}$ be a finite forest of depth $d$.
		If it is possible, we build a forest $\forest{X}$ such that $\cut{\unrollP{\forest{X}}}{d} = \forest{F}$ in polynomial time.
	\end{Theorem}

\bibliography{biblio}

\end{document}